\documentclass[12pt]{iopart}
\usepackage{amsthm,graphicx}
\font\bb=msbm10 at 12pt
\newcommand{\mathbb}[1]{\hbox{\bb #1}} %realna, prirozena cisla, atd.

\newtheorem{theorem}{Theorem}[section]
\newtheorem{lemma}{Lemma}[section]
\begin{document}

\title[Perturbed quantum graphs with rationally related edges]
{Resonances from perturbations of quantum graphs with~rationally related edges}

\author{Pavel Exner}
\address{Doppler Institute for Mathematical Physics and Applied
Mathematics, \\ Czech Technical University, B\v{r}ehov\'{a} 7,
11519 Prague, \\ and  Nuclear Physics Institute ASCR, 25068
\v{R}e\v{z} near Prague, Czechia} \ead{exner@ujf.cas.cz}

\author{Ji\v{r}\'{\i} Lipovsk\'{y}}
\address{Institute of Theoretical Physics, Faculty of Mathematics and
Physics, \\ Charles University, V Hole\v{s}ovi\v{c}k\'ach 2, 18000
Prague, \\ and  Nuclear Physics Institute ASCR, 25068
\v{R}e\v{z} near Prague, Czechia} \ead{lipovsky@ujf.cas.cz}

\begin{abstract}
We discuss quantum graphs consisting of a compact part and
semiinfinite leads. Such a system may have embedded eigenvalues if
some edge lengths in the compact part are rationally related. If
such a relation is perturbed these eigenvalues may turn into
resonances; we analyze this effect both generally and in simple
examples.
\end{abstract}

\maketitle

%%%%
%%%% Introduction
%%%%
\section{Introduction}

Quantum graphs have attracted a lot of attention recently. The
reason is not only that they represent a suitable model for
various microstructures, being thus of a direct practical value,
but also that they are an excellent laboratory to study a variety
of quantum effects. This comes from a combination of two features.
On one hand these models are mathematically accessible since the
objects involved are ordinary differential operators. On the other
hand graphs may exhibit a rich geometrical and topological
structure which influences behaviour of quantum particle for which
such a graph is a configuration space. There is nowadays a huge
literature on quantum graphs and, instead of presenting a long
list of references, we restrict ourselves to mentioning the review
papers \cite{Ku, Ku2} as a guide to further reading.

One important property of quantum graphs is that --- in contrast
to usual Schr\"odinger operators --- the unique continuation
principle is in general not valid for them: they can exhibit
eigenvalues with compactly supported eigenfunctions even if the
graph extends to infinity. This property is closely connected with
the fact that eigenvalues embedded in the continuous spectrum are
on quantum graphs by far less exceptional then for usual
Schr\"odinger operators. A typical situation when this happens is
when the graph contains a loop consisting of edges with rationally
related lengths and the eigenfuction has zeros at the
corresponding vertices, which prevents it from ``communicating''
with the rest of the graph.

On the other hand, since such an effect leans on rational
relations between the edge lengths, it is unstable with respect to
perturbations which change these ratios. The resolvent poles
associated with the embedded eigenvalues do not disappear under
such a geometric perturbation, though, and one can naturally
expect that they move into the second sheet of the complex energy
surface producing resonances. The aim of the present paper is to
discuss this effect in a reasonably general setting.

We consider a graph consisting of a compact ``inner''part to which
a finite number of semiifinite leads are attached. We assume a
completely general coupling of wavefunctions at the graph vertices
consistent with the self-adjointness requirement. As a preliminary
we will show, generalizing the result of \cite{EL}, that we can
speak about resonances without further adjectives because the
resolvent and scattering resonances coincide in the present case.
We also show how the problem can be rephrased on the compact graph
part only by introducing an effective, energy-dependent coupling.

After that we formulate general conditions under which such a
quantum graph possesses embedded eigenvalues in terms of the graph
geometry (edge lengths) and the matrix of coupling parameters. The
discussion of the behaviour of embedded eigenvalues is opened by a
detailed analysis of two simple examples, those of a ``loop'' and
a ``cross'' resonator graphs. Here we can analyze not only the
effect of small length perturbations but also, using numerical
solutions, to find the global pole behaviour and to illustrate
several different types of it. Returning to the general analysis
in the closing section, we will derive conditions under which the
eigenvalues remain embedded, and show that ``nothing is lost at
the perturbation'' in the sense that the number of poles,
multiplicity taken into account, is preserved.

%%%%
%%%% Preliminaries
%%%%
\section{Preliminaries}\label{sec2}

\subsection{A universal setting for graphs with leads}
\label{flower}

Let us consider a graph $\Gamma$ consisting of a set of vertices
$\mathcal{V}=\{\mathcal{X}_j: j\in I\}$, a set of finite edges
$\mathcal{L} =\{\mathcal{L}_{jn}:  (\mathcal{X}_j, \mathcal{X}_n)
\in I_\mathcal{L} \subset I\times I\}$ and a set of infinite edges
${\mathcal L_\infty} = \{\mathcal{L}_{j\infty}:  \mathcal{X}_j \in
I_\mathcal{C}\}$ attached to them. We regard it as a configuration
space of a quantum system with the Hilbert space
 % ------------- %
 $$
   \mathcal{H} = \bigoplus_{L_{j} \in \mathcal{L}} L^2([0,l_{j}])\oplus
\bigoplus_{\mathcal{L}_{j\infty} \in \mathcal{L_\infty}}
L^2([0,\infty)).
 $$
 % ------------- %
the elements of which can be written as columns $\psi  = (f_{j}:
\mathcal{L}_j \in \mathcal{L},\, g_{j}: \mathcal{L}_{j\infty}\in
\mathcal{L}_\infty)^T$. We consider the dynamics governed by a
Hamiltonian which acts as $-\mathrm{d}^2/\mathrm{d} x^2$ on each
link. In order to make it a self-adjoint operator, boundary
conditions
 % ------------- %
 \begin{equation}
   (U_j -I) \Psi_j +i (U_j +I) \Psi_j' = 0\,\label{rat-coup1}
 \end{equation}
 % ------------- %
with unitary matrices $U_j$ have to be imposed  at the vertices
$\mathcal{X}_j$, where $\Psi_j$ and $\Psi_j'$ are vectors of the
functional values and of the (outward) derivatives at the
particular vertex, respectively. In other words, the domain of the
Hamiltonian consists of all functions in
$W^{2,2}(\mathcal{L}\oplus\mathcal{L}_\infty)$ which satisfy the
conditions (\ref{rat-coup1}). We will speak about the described
structure as of a \emph{quantum graph} and as long as there is no
danger of misunderstanding we will use for simplicity the symbol
$\Gamma$ again.

While the model is simple dealing with a complicated graph may be
nevertheless cumbersome. To make it easier we will employ a trick
mentioned to our knowledge for the first time in \cite{Ku2}
passing to a graph $\Gamma_0$ in which all edge ends meet
in a single vertex as sketched in Fig.~\ref{fig_flower}; the
actual topology of $\Gamma$ will be then encoded into the matrix
which describes the coupling in the vertex.

To be more specific, suppose that $\Gamma$ described above has an
adjacency matrix $C_{ij}$ and that matrices $U_j$ describe the
coupling between vectors of functional values $\Psi_j$ and
derivatives $\Psi_j'$ at $\mathcal{X}_j$. This will correspond to
the ``flower-like'' graph with one vertex, the set of loops
isomorphic to $\mathcal{L}$ and the set of semiinfinite links
$\mathcal{L}_\infty$ which does not change; coupling at the only
vertex of this graph is given by a ``big'' unitary matrix $U$.

Denoting $N=\mathrm{card}\,\mathcal{L}$ and
$M=\mathrm{card}\,\mathcal{L_\infty}$ we introduce the
$(2N+M)$-dimensional vector of functional values by $\Psi =
(\Psi_1^T, \dots ,\Psi_{\mathrm{card}\,\mathcal{V}}^T)^T$ and
similarly the vector of derivatives $\Psi'$ at the vertex. The
valency of this vertex is $M +\sum_{i,j} C_{ij}= 2N+M$. One can
easily check that the conditions (\ref{rat-coup1}) can be
rewritten on $\Gamma_0$ using one $(2 N+ M)\times (2 N+ M)$
unitary block diagonal matrix $U$ consisting of blocks $U_j$ as
 % ------------- %
 \begin{equation}
   (U -I) \Psi +i (U+I) \Psi' = 0\,;\label{rat-coup2}
 \end{equation}
 % ------------- %
the equation (\ref{rat-coup2}) obviously decouples into the set of
equations (\ref{rat-coup1}) for $\Psi_j$ and $\Psi'_j$.

Since neither the edge lengths and the corresponding Hilbert
spaces nor the operator action on them are affected and the only
change is a possible edge renumbering the quantum graph $\Gamma_0$
is related to the original $\Gamma$ by the natural unitary
equivalence and the spectral properties we are interested in are
not affected by the model modification.

\begin{figure}
   \begin{center}
     \includegraphics{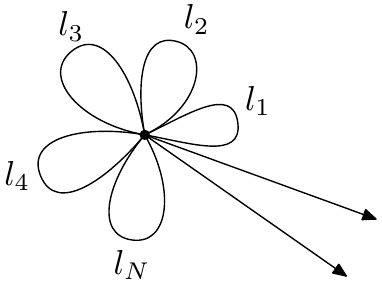}
     \caption{The model $\Gamma_0$ for a quantum graph $\Gamma$ with $N$
     internal finite edges and $M$ external links}\label{fig_flower}
   \end{center}
\end{figure}

\subsection{Equivalence of the scattering and resolvent resonances} \label{equiv}

As another preliminary we need a few facts about resonances on
quantum graphs. In \cite{EL} we studied the situation where to
each vertex of a compact graph at most one external semi-infinite
link is attached; we have demonstrated that the resonances may be
equivalently understood as poles of the analytically continued
resolvent, $(H - \lambda \,\mathrm{id})^{-1}$, or of the on-shell
scattering matrix. Here we extend the result to all quantum graphs
with finite number of edges, both finite and semi-infinite: we
will show that the resolvent and scattering resonances again
coincide.The above described ``flower-like'' graph model allows us
to give an elegant proof of this claim.

Let us begin with the resolvent resonances. As in \cite{EL} the
idea is to employ an exterior complex scaling; this seminal idea
can be traced back to the work J.-M.~Combes and coauthors,
cf.~\cite{AC}, and its use in the graph setting is particularly
simple. Looking for complex eigenvalues of the scaled operator we
do not change the compact-graph part: using the Ansatz $f_j (x) =
a_j \sin{kx} + b_j \cos{kx}$ on the internal edges we obtain
 % ------------- %
 \begin{eqnarray}
    f_j (0) = b_j, &\quad &f_j (l_j)= a_j \sin{kl_j}
    + b_j \cos{kl_j},\label{rat-fce}\\
    f_j' (0) = k a_j, &\quad -&f_j' (l_j)= - k a_j \cos{kl_j}
    + k b_j \sin{kl_j},\label{rat-der}
 \end{eqnarray}
 % ------------- %
hence we have
 % ------------- %
 \begin{eqnarray}
   \left(\begin{array}{c}f_j(0)\\ f_j(l_j)\end{array}\right) &=
   \left(\begin{array}{cc}0 & 1\\ \sin{kl_j}&\cos{kl_j}\end{array}\right)
   \left(\begin{array}{c}a_j\\ b_j\end{array}\right)\,,\label{rat-int1}
   \\
   \left(\begin{array}{c}f_j'(0)\\ -f_j'(l_j)\end{array}\right) &=
   k \left(\begin{array}{cc}1 & 0\\ -\cos{kl_j}&\sin{kl_j}\end{array}\right)
   \left(\begin{array}{c}a_j\\ b_j\end{array}\right)\,. \label{rat-int2}
 \end{eqnarray}
 % ------------- %
On the other hand, the functions on the semi-infinite edges are
scaled by $g_{j\theta}(x) = \mathrm{e}^{\theta/2} g_j(x
\mathrm{e}^\theta)$ with an imaginary~$\theta$ rotating the
essential spectrum of the transformed (non-selfadjoint)
Hamiltonian into the lower complex halfplane so that the poles of
the resolvent on the second sheet become ``uncovered'' if the
rotation angle is large enough. The argument is standard, both generally
and in the graph setting \cite{EL}, so we skip the details. In
particular, the ``exterior'' boundary values are given by
 % ------------- %
 \begin{equation}
   g_j (0) = \mathrm{e}^{-\theta/2} g_{j\theta},\quad g_j' (0)
   = ik \mathrm{e}^{-\theta/2} g_{j\theta}.\label{rat-ext}
 \end{equation}
 % ------------- %

Now we substitute eqs. (\ref{rat-int1}), (\ref{rat-int2}) and
(\ref{rat-ext}) into (\ref{rat-coup2}). We rearrange the terms in
$\Psi$ and $\Psi'$ in such a way that the functional values
corresponding to the two ends of each edge are neighbouring, and
the entries of the matrix $U$ are rearranged accordingly. This
yields
 % ------------- %
 \begin{equation}
   (U-I)C_1(k)\left(\begin{array}{c}a_1\\b_1\\a_2\\\vdots\\b_N\\
   \mathrm{e}^{-\theta/2}g_{1\theta}\\\vdots\\
   \mathrm{e}^{-\theta/2}g_{M\theta}\end{array}\right)
   +ik(U+I)C_2(k)\left(\begin{array}{c}a_1\\b_1\\a_2\\\vdots\\b_N\\
   \mathrm{e}^{-\theta/2}g_{1\theta}\\\vdots\\
   \mathrm{e}^{-\theta/2}g_{M\theta}\end{array}\right)=0,\label{rat-scal}
 \end{equation}
 % ------------- %
where the matrices $C_1$, $C_2$ are given by $C_1 (k)  = \mathrm{diag\,}(C_1^{(1)}(k),C_1^{(2)}(k),\dots,C_1^{(N)}(k), I_{M\times M})$ and $C_2  = \mathrm{diag\,}(C_2^{(1)}(k),C_2^{(2)}(k),\dots,C_2^{(N)}(k), i I_{M\times M})$, respectively, where
 % ------------- %
 \begin{eqnarray*}
   C_1^{(j)} (k) =
   \left(\begin{array}{cc}
   0 & 1 \\
   \sin{kl_j} & \cos{kl_j}\end{array}\right)\,,\qquad
   C_2^{(j)} (k) =
   \left(\begin{array}{cc}
   1 & 0  \\
   -\cos{kl_j} & \sin{kl_j}\end{array}\right)
 \end{eqnarray*}
 % ------------- %
and $I_{M\times M}$ is a $M\times M$ unit matrix.

The solvability condition of the system (\ref{rat-scal})
determines the eigenvalues of scaled non-selfadjoint operator, and
\emph{mutatis mutandis} the poles of the analytically continued
resolvent of the original graph Hamiltonian.

The other standard approach to resonances is to study poles of the
on-shell scattering matrix, again in the lower complex halfplane.
In our particular case we choose a combination of two planar
waves, $g_j = c_j \mathrm{e}^{-ikx}+ d_j \mathrm{e}^{ikx}$, as an
Ansatz on the external edges; we ask about poles of the matrix
$S=S(k)$ which maps the vector of amplitudes of the incoming waves
$c=\{c_n\}$ into the vector of the amplitudes of the outgoing
waves $d =\{d_n\}$ by $d=Sc$. The condition for the scattering
resonances is then $\mathrm{det}\,S^{-1}=0$ for appropriate
complex values of $k$. The functional values and derivatives at
the vertices are now given by
 % ------------- %
 $$
   g_j (0) = c_j + d_j,\quad g_j' (0) = i k (d_j - c_j).
 $$
 % ------------- %
together with eqs. (\ref{rat-fce})--(\ref{rat-der}). After
substituting into (\ref{rat-coup2}) one arrives at the condition
 % ------------- %
 $$
   (U-I)C_1(k)\left(\begin{array}{c}a_1\\b_1\\a_2\\
     \vdots\\b_N\\c_1+d_1\\\vdots\\ c_M+d_M\end{array}\right)
   +ik(U+I)C_2(k)\left(\begin{array}{c}a_1\\b_1\\a_2\\
     \vdots\\b_N\\d_1-c_1\\\vdots\\ d_M-c_M\end{array}\right)=0.
 $$
 % ------------- %
Since we are interested in zeros of $\mathrm{det}\,S^{-1}$, we
regard the previous relation as an equation for variables $a_j$,
$b_j$ and $d_j$ while $c_j$ are just parameters, in other words
 % ------------- %
 $$
   [(U-I)C_1(k) +ik(U+I)C_2(k)]\left(\begin{array}{c}a_1\\b_1\\a_2\\
     \vdots\\b_N\\d_1\\\vdots\\ d_M\end{array}\right)
   =[-(U-I)C_1(k)+ik(U+I)C_2(k)]\left(\begin{array}{c}0\\0\\
     \vdots\\0\\c_1\\\vdots\\c_M\end{array}\right)\,.
 $$
 % ------------- %
Eliminating the variables $a_j$, $b_j$ one can derive from here a
system of $M$ equations expressing the map $S^{-1}d =c$. The
condition under which the previous system is not solvable, what is
equal to $\mathrm{det}\,S^{-1}=0$, reads
 % ------------- %
 \begin{equation}
   \mathrm{det}\,\left[(U-I)\,C_1(k)+ik (U+I)\,C_2(k)\right]=0 \label{rat-res}
 \end{equation}
 % ------------- %
being the same as the condition of solvability of the system
(\ref{rat-scal}); this means that the families of resonances
determined in the two ways coincide.

\subsection{Effective coupling on the finite graph}

The study of resonances can be further simplified by reducing it
to a problem on the compact subgraph only. The idea is to replace
the coupling at the vertex where external semi-infinite edges are
attached by an effective one obtained by eliminating the external
variables. Substituting from (\ref{rat-ext}) into eqs.
(\ref{rat-coup2}) we get
 % ------------- %
 \begin{equation}
  \hspace{-17mm}(U-I)\left(\begin{array}{c}f_1\\\vdots\\f_{2N}\\\mathrm{e}^{-\theta/2}g_{1\theta}\\
  \vdots\\\mathrm{e}^{-\theta/2}g_{M\theta}\end{array}\right)
  + (U+I)\,\mathrm{diag}(i,\dots,i,-k,\dots,-k) \left(\begin{array}{c}f_1'\\
  \vdots\\f_{2N}'\\\mathrm{e}^{-\theta/2}g_{1\theta}\\\vdots\\\mathrm{e}^{-\theta/2}g_{M\theta}
  \end{array}\right)=0\,.\label{rat-ffg}
 \end{equation}
 % ------------- %
We consider now $U$ as a matrix consisting of four blocks,
$U=\left(\begin{array}{cc}U_1 & U_2 \\U_3 & U_4
\end{array}\right)$, where $U_1$ is the $2N\times 2N$ square
matrix referring to the compact subgraph, $U_4$ is the $M\times M$
square matrix related to the exterior part, and $U_2$ and $U_3$
are rectangular matrices of the size $M \times 2N$ and $2N \times
M$, respectively, connecting the two. Then the previous set of
equations can be written as
 % ------------- %
 $$
    V (f_1,\dots,f_{2N},f_1',\dots,f_{2N}',\mathrm{e}^{-\theta/2}g_{1\theta},
    \dots,\mathrm{e}^{-\theta/2}g_{M\theta})^\mathrm{T} = 0,
 $$
 % ------------- %
where
 % ------------- %
 $$
   V = \left(\begin{array}{ccc}U_1 -I & i (U_1 + I)& (1-k) U_2\\
   U_3 & i U_3 & (1-k) U_4 - (k+1)I\end{array}\right)\,.
 $$
 % ------------- %
If the matrix $[(1-k)U_4-(k+1)]$ is regular, one obtains from here
 % ------------- %
 $$
   (\mathrm{e}^{-\theta/2}g_{1\theta},\dots,\mathrm{e}^{-\theta/2}g_{M\theta})^\mathrm{T}
   = -[(1-k)U_4 - (k+1) I]^{-1}U_3 (f_1+if_1',\dots,f_{2N}+if_{2N}')^\mathrm{T}
 $$
 % ------------- %
and substituting it further into (\ref{rat-ffg}) we find that the
following expression,
 % ------------- %
 \begin{eqnarray*}
   \left\{U_1-I -(1-k)U_2 [(1-k)U_4-(k+1)I]^{-1}U_3\right\}(f_1,\dots,f_{2N})^\mathrm{T}+
\\
   +i\left\{U_1+I-(1-k)U_2
   [(1-k)U_4-(k+1)I]^{-1}U_3\right\}(f_1',\dots,f_{2N}')^\mathrm{T}=0\,.
 \end{eqnarray*}
 % ------------- %
must vanish. Consequently, elimination of the external part leads
to an effective coupling on the compact part of the graph
expressed by the condition
 % ------------- %
 $$
   (\tilde U(k)-I)(f_1,\dots,f_{2N})^\mathrm{T}
   + i (\tilde U(k)+I)(f_1',\dots,f_{2N}')^\mathrm{T} = 0\,,
 $$
 % ------------- %
where the corresponding coupling matrix
 % ------------- %
 \begin{equation} \label{en-dep}
   \tilde U(k)=U_1 -(1-k) U_2 [(1-k) U_4 - (k+1) I]^{-1} U_3 \label{rat-efu}
 \end{equation}
 % ------------- %
is obviously energy-dependent and, in general, may not be unitary.

%%%%
%%%% Embedded eigenvalues for graphs with rationally related edges
%%%%
\section{Embedded eigenvalues for graphs with rationally related edges}

As mentioned in the introduction, quantum graphs of the type we
consider here have the positive halfline as the essential
spectrum, and they may have eigenvalues with compactly supported
eigenfunctions embedded in it.

\subsection{A general result}

We will focus on graphs which contain several internal edges of
lengths equal to integer multiples of a fixed $l_0 > 0$. In the
spirit of the previous section we restrict ourselves only to
compact graphs remembering that the presence of an exterior part
can be rephrased through an effective energy-dependent coupling
replacing the original $U$ by the matrix $\tilde U(k)$ defined
above.

Following Sec.~\ref{flower} we model a given compact $\Gamma$ by
$\Gamma_0$ having only one vertex and $N$ finite edges emanating
from this vertex and ending at it. The coupling between the edges
is described by a $2N \times 2N$ unitary matrix $U$ and condition
(\ref{rat-coup2}). Suppose that the lengths of the first $n$ edges
are integer multiples of a positive real number $l_0$. Our aim is
to find out for which matrices $U$ the spectrum of the
corresponding Hamiltonian $H=H_U$ contains the eigenvalues $k = 2m
\pi /l_0$, $m\in \mathbb{N}$.

Since our graph is not directed it is convenient to work in a setting
invariant with respect to interchange of the edge ends. To this
aim we choose the Ansatz
 % ------------- %
 $$
   \Psi_j (x) = A_j \sin{k(x - l_j/2)}+ B_j \cos{k(x - l_j/2)}\,.
 $$
 % ------------- %
on the $j$-th edge. Subsequently, one gets
 % ------------- %
 \begin{eqnarray*}
   \left(\begin{array}{c}\Psi_j(0)\\\Psi_j(l_j)\end{array}\right) &=&
   \left(\begin{array}{cc}-\sin{\frac{kl_j}{2}}&\cos{\frac{kl_j}{2}}\\
   \sin{\frac{kl_j}{2}}&\cos{\frac{kl_j}{2}}\end{array}\right)
   \left(\begin{array}{c}A_j \\ B_j\end{array}\right)\,,
\\
   \left(\begin{array}{c}\Psi_j'(0)\\-\Psi_j'(l_j)\end{array}\right) &=&
   k \left(\begin{array}{cc}\cos{\frac{kl_j}{2}}&\sin{\frac{kl_j}{2}}\\
   -\cos{\frac{kl_j}{2}}&\sin{\frac{kl_j}{2}}\end{array}\right)
   \left(\begin{array}{c}A_j \\ B_j\end{array}\right)\,.
 \end{eqnarray*}
 % ------------- %
The eigenvalue condition, expressed in terms of solvability of the
system (\ref{rat-coup2}), is given by
 % ------------- %
 \begin{equation}
   \mathrm{det}\,[U D_1(k) + D_2(k)] = 0\,, \label{rat-con-l2}
 \end{equation}
 % ------------- %
where
 % ------------- %
 {\scriptsize
 $$
   D_1(k) = \left(\begin{array}{ccccc}
   -\sin{\frac{kl_1}{2}}+ik\cos{\frac{kl_1}{2}} & \cos{\frac{kl_1}{2}}+ik\sin{\frac{kl_1}{2}} &  \cdots &0 &0\\
   \sin{\frac{kl_1}{2}}-ik\cos{\frac{kl_1}{2}} & \cos{\frac{kl_1}{2}}+ik\sin{\frac{kl_1}{2}} &  \cdots &0 &0\\
%   0&0&-\sin{\frac{kl_2}{2}}+ik\cos{\frac{kl_1}{2}} &\cdots&0 &0
   \vdots &\vdots &\ddots &\vdots &\vdots\\
   0 &0 & \cdots& -\sin{\frac{kl_N}{2}}+ik\cos{\frac{kl_N}{2}} & \cos{\frac{kl_N}{2}}+ik\sin{\frac{kl_N}{2}} \\
   0 &0 & \cdots&\sin{\frac{kl_N}{2}}-ik\cos{\frac{kl_N}{2}} & \cos{\frac{kl_N}{2}}+ik\sin{\frac{kl_N}{2}}
   \end{array}\right)\,,
 $$
 $$
   D_2(k) = \left(\begin{array}{ccccc}
   \sin{\frac{kl_1}{2}}+ik\cos{\frac{kl_1}{2}} & -\cos{\frac{kl_1}{2}}+ik\sin{\frac{kl_1}{2}} &  \cdots &0 &0\\
   -\sin{\frac{kl_1}{2}}-ik\cos{\frac{kl_1}{2}} & -\cos{\frac{kl_1}{2}}+ik\sin{\frac{kl_1}{2}} &  \cdots &0 &0\\
%   0&0&-\sin{\frac{kl_2}{2}}+ik\cos{\frac{kl_1}{2}} &\cdots&0 &0
   \vdots &\vdots &\ddots &\vdots &\vdots\\
   0 &0 & \cdots& \sin{\frac{kl_N}{2}}+ik\cos{\frac{kl_N}{2}} & -\cos{\frac{kl_N}{2}}+ik\sin{\frac{kl_N}{2}} \\
   0 &0 & \cdots&-\sin{\frac{kl_N}{2}}-ik\cos{\frac{kl_N}{2}} & -\cos{\frac{kl_N}{2}}+ik\sin{\frac{kl_N}{2}}
   \end{array}\right)\,.
 $$

 }
 % ------------- %
For a future purpose, let us rewrite the spectral condition
(\ref{rat-con-l2}) in the form $\mathrm{det}\,(C(k)+S(k))=0$,
where the matrix $C(k)$ contains terms with $\cos{\frac{kl_j}{2}}$
and $S(k)$ contains those with $\sin{\frac{kl_j}{2}}$. Hence all
the entries in the first $2n$ columns of $S(k)$ vanish for $k=
2m\pi/l_0$, $m\in \mathbb{N}$ while the others can be nontrivial.
Similarly, all the entries in the first $2n$ columns of $C(k)$ are
for $k = (2m+1)\pi/l_0$, $m\in \mathbb{N}$ equal to zero. The
entries of the ``cosine'' matrix are
 % ------------- %
 \begin{eqnarray*}
   C_{i,2j-1}(k) = (u_{i,2j-1} - u_{i,2j}) ik \cos{\frac{kl_j}{2}}+
   (\delta_{i,2j-1}-\delta_{i,2j})ik \cos{\frac{kl_j}{2}}\,,\\
   C_{i,2j}(k) = (u_{i,2j-1} + u_{i,2j}) \cos{\frac{kl_j}{2}}-
   (\delta_{i,2j-1}+\delta_{i,2j})\cos{\frac{kl_j}{2}}\,.
 \end{eqnarray*}
 % ------------- %
Similarly, the entries of $S(k)$ are
 % ------------- %
 \begin{eqnarray*}
   S_{i,2j-1}(k) = (-u_{i,2j-1} + u_{i,2j}) \sin{\frac{kl_j}{2}}+
   (\delta_{i,2j-1}-\delta_{i,2j})\sin{\frac{kl_j}{2}}\,,\\
   S_{i,2j}(k) = (u_{i,2j-1} + u_{i,2j})  ik \sin{\frac{kl_j}{2}}+
   (\delta_{i,2j-1}+\delta_{i,2j})ik\sin{\frac{kl_j}{2}}\,.
 \end{eqnarray*}
 % ------------- %

First of all, let us consider the situation when $\sin{kl_0/2}=0$.

\begin{theorem}\label{veta1}
Let a graph $\Gamma_0$ consist of a single vertex and $N$ finite
edges emanating from this vertex and ending at it, and suppose
that the coupling between the edges is described by a $2N \times
2N$ unitary matrix $U$ and condition (\ref{rat-coup2}). Let the
lengths of the first $n$ edges be integer multiples of a positive real
number $l_0$. If the rectangular $2N \times 2n$ matrix
% ------------- %
 \begin{equation}
   \hspace{-18mm}M_{\mathrm{even}} = \left(\begin{array}{ccccccc}
   u_{11}& u_{12}-1&u_{13}&u_{14}&\cdots&u_{1,2n-1}&u_{1,2n}\\
   u_{21}-1& u_{22}&u_{23}&u_{24}&\cdots&u_{2,2n-1}&u_{2,2n}\\
   u_{31}& u_{32}&u_{33}&u_{34}-1&\cdots&u_{3,2n-1}&u_{3,2n}\\
   u_{41}& u_{42}&u_{43}-1&u_{44}&\cdots&u_{4,2n-1}&u_{4,2n}\\
   \vdots&\vdots&\vdots&\vdots&\ddots&\vdots&\vdots\\
   u_{2N-1,1}& u_{2N-1,2}&u_{2N-1,3}&u_{2N-1,4}&\cdots&u_{2N-1,2n-1}&u_{2N-1,2n}\\
   u_{2N,1}& u_{2N,2}&u_{2N,3}&u_{2N,4}&\cdots&u_{2N,2n-1}&u_{2N,2n}
   \end{array}\right)\label{rat-matice}
 \end{equation}
 % ------------- %
has rank smaller than $2n$ then the spectrum of the corresponding
Hamiltonian $H=H_U$ contains eigenvalues of the form $\epsilon = 4 m^2 \pi^2
/l_0^2$ with $m\in\mathbb{N}$ and the multiplicity of these
eigenvalues is at least the difference between $2n$ and the rank of
$M_{\mathrm{even}}$.
\end{theorem}
\begin{proof}
The condition (\ref{rat-con-l2}) is clearly satisfied if the
rectangular matrix containing only the first $2n$ columns has rank
smaller than $2n$, because then some of the columns of matrix
$C(k)+ S(k)$ are linearly dependent. Since all the entries of the
first $2n$ columns of $S(k)$ contain the term $\sin{kl_j/2}$,
which disappear for $kl_0 = 2m \pi$, one can consider the matrix
$C(k)$ only. Dividing some of the columns of $C(k)$ by appropriate
nonzero terms, which is possible since $\cos{kl_j/2}\not = 0$ for
$\sin{kl_0/2} = 0$, and subtracting them from each other does not
change the rank of the matrix. This argument shows that the rank
of matrix $M_{\mathrm{even}}$ must be smaller than $2n$ in order
to yield a solution of the condition (\ref{rat-con-l2}) and that
the multiplicity is given by the difference.
\end{proof}

It is important to notice that the unitarity of $U$ played no role
in the argument, and consequently, one can obtain in this way
embedded eigenvalues $\epsilon = 4 m^2 \pi^2 /l_0^2$ for a graph
containing external links, however, the matrix
$M_{\mathrm{even}}(k)$ defined in analogy with (\ref{rat-matice})
must have rank smaller than $2n$ for \emph{all} values of $k$.

Mathematically speaking the described case does not involve only
cases where the original graph $\Gamma$ contains a loop with
rational rate of the lengths of the edges. Choosing appropriate
$U$ one can find such eigenvalues also for graphs where the edges
of $\Gamma$ with lengths equal to integer multiples of $l_0$ are
not adjacent. This corresponds, however, to couplings allowing the
particle to ``hop'' between different vertices which is not so
interesting from the point of view of the underlying physical
model.

A similar claim can be made also for $kl_0$ equal to odd multiples
of $\pi$.
\begin{theorem}\label{veta2}
If under the same assumptions as above, the rectangular $2N \times
2n$ matrix
% ------------- %
 \begin{equation}
   \hspace{-18mm}M_{\mathrm{odd}} = \left(\begin{array}{ccccccc}
   u_{11}& u_{12}+1&u_{13}&u_{14}&\cdots&u_{1,2n-1}&u_{1,2n}\\
   u_{21}+1& u_{22}&u_{23}&u_{24}&\cdots&u_{2,2n-1}&u_{2,2n}\\
   u_{31}& u_{32}&u_{33}&u_{34}+1&\cdots&u_{3,2n-1}&u_{3,2n}\\
   u_{41}& u_{42}&u_{43}+1&u_{44}&\cdots&u_{4,2n-1}&u_{4,2n}\\
   \vdots&\vdots&\vdots&\vdots&\ddots&\vdots&\vdots\\
   u_{2N-1,1}& u_{2N-1,2}&u_{2N-1,3}&u_{2N-1,4}&\cdots&u_{2N-1,2n-1}&u_{2N-1,,2n}\\
   u_{2N,1}& u_{2N,2}&u_{2N,3}&u_{2N,4}&\cdots&u_{2N,2n-1}&u_{2N,2n}
   \end{array}\right)\,,
 \end{equation}
 % ------------- %
has rank smaller than $2n$ then the spectrum of the corresponding
Hamiltonian $H=H_U$ contains eigenvalues of the form $\epsilon = (2m+1)^2
\pi^2 /l_0^2$ with $m\in\mathbb{N}$ and the multiplicity of these
eigenvalues is at least the difference between $2n$ and rank of
$M_{\mathrm{odd}}$.
\end{theorem}

We skip the proof which is similar to the previous one, the change
being that the roles of the matrices $S(k)$ and $C(k)$ are
interchanged. We also notice that similarly as above the results
extends to graphs with semi-infinite external edges.

\subsection{A loop with $\delta$ or $\delta'_\mathrm{s}$  couplings}

\begin{figure}
   \begin{center}
     \includegraphics{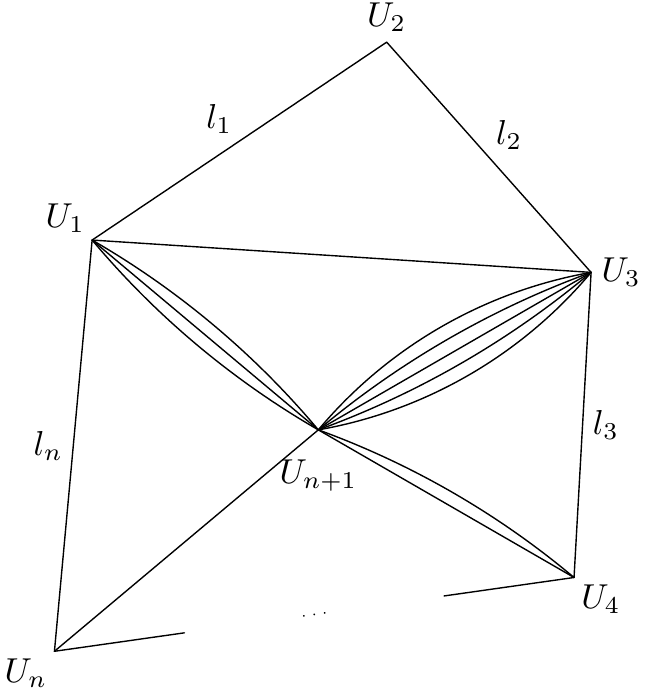}
     \caption{A loop of the edges with rational rate of their lengths}\label{fig-loopdelta}
   \end{center}
\end{figure}

As mentioned above a prime example of embedded eigenvalues in the
considered class of quantum graphs concerns the situation when $\Gamma$
contains a subgraph in the form of a loop of $n$ edges with the
lengths equal to integer multiples of $l_0$. We denote by $U_j,\,
j=1,\dots n$, the unitary matrices describing the coupling at the
vertices of such a loop and by $U_{n+1}$ the unitary matrix which
describes the coupling at all the other vertices of the graph ---
cf.~Fig.~\ref{fig-loopdelta}. The unitary matrix which describes
the coupling on the whole graph, in the sense explained in
Sec.~\ref{flower}, is
 % ------------- %
 $$
    U = \left(\begin{array}{cccc} U_1& 0 & \cdots & 0\\
    0 & U_2 & \cdots & 0\\
    \vdots & \vdots & \ddots & \vdots\\
    0 & 0 & \cdots & U_{n+1}
    \end{array}\right)\,.
 $$
 % ------------- %
Let us further restrict our attention to the case when the
coupling in the loop vertices is \emph{invariant with respect
to the permutation of edges}, i.e. suppose that matrices $U_1$,
\dots, $U_n$ can be written as $U_j = a_j J + b_j I$, where $I$
is a unit matrix, $J$ is a matrix with all entries equal to one
and $a_j$ and $b_j$ are complex numbers satisfying $|b_j|=1$ and
$|b_j+a_j\mathrm{deg\,}\mathcal{X}_j|=1$ to make the operator
self-adjoint --- cf.~\cite{ET}.

Recall that in order to use Theorems~\ref{veta1} and \ref{veta2}
one has to rearrange the columns and rows of the unitary matrix
$U$ accordingly. The first $2n$ entries in $\Psi$ and $\Psi'$
correspond to the edges with rational rates of their lengths.
Therefore, appropriate permutations of columns and rows of $U$
must be performed: the first two columns should correspond to the first
edge of the loop (from the vertex 1 to 2), the second two columns
to the second edge, etc. The rearranged coupling matrix is thus
$\left(\begin{array}{cc}U_{\mathrm{r}}&0\\0&U_{n+1}\end{array}\right)$
with

 % ------------- %
  {\scriptsize
 $$
    U_{\mathrm{r}} = \left(\begin{array}{c@{}c@{}c@{}c@{}c @{}c@{}c
           @{}c@{}c@{}c@{}c@{}c@{}c@{}c@{}c@{}c}
    a_1 +b_1 & 0 & 0& \cdots & 0 & 0 & a_1& a_1 & \cdots & a_1& 0 &
      \cdots & 0& 0 & \cdots & 0\\
    0 &  a_2 +b_2 & a_2 & \cdots &0 & 0 & 0& 0 & \cdots & 0& a_2 &
      \cdots & a_2& 0 & \cdots & 0\\
    0 & a_2 & a_2+ b_2\,\,& \cdots &0 & 0 & 0& 0 & \cdots & 0& a_2 &
      \cdots & a_2& 0 & \cdots & 0\\
    \vdots & \vdots& \vdots&  \ddots &\vdots &\vdots & \vdots& \vdots&
      \ddots &\vdots & \vdots&  \ddots &\vdots & \vdots&  \ddots &\vdots \\
    0 & 0 & 0 &  \cdots &\,\,a_n+b_n & a_n&0& 0 & \cdots & 0& 0 &
      \cdots & 0& a_n & \cdots & a_n\\
    0 & 0 & 0 &  \cdots &a_n & a_n+b_n&0& 0 & \cdots & 0& 0 &
      \cdots & 0& a_n & \cdots & a_n\\
    a_1 & 0 & 0 & \cdots &0 &0& a_1+b_1& a_1 & \cdots & a_1& 0 &
      \cdots & 0& 0 & \cdots & 0\\
    a_1 &  0 & 0&  \cdots &0 & 0 & a_1& a_1+b_1\,\, & \cdots & a_1& 0 &
      \cdots & 0& 0 & \cdots & 0\\
   \vdots & \vdots& \vdots&  \ddots &\vdots &\vdots & \vdots& \vdots&
      \ddots &\vdots & \vdots&  \ddots &\vdots & \vdots&  \ddots &\vdots \\
    a_1 &  0 & 0&  \cdots &0 & 0 & a_1& a_1 & \cdots & \,\,a_1+b_1& 0 &
      \cdots & 0& 0 & \cdots & 0\\
    0 & a_2 & a_2&\cdots & 0 & 0 & 0& 0 & \cdots & 0& a_2+b_2 \,\,&
      \cdots & a_2& 0 & \cdots & 0\\
    \vdots & \vdots& \vdots&  \ddots &\vdots &\vdots & \vdots& \vdots&
      \ddots &\vdots & \vdots&  \ddots &\vdots & \vdots&  \ddots &\vdots \\
    0 & a_2 & a_2&\cdots & 0 & 0 & 0& 0 & \cdots & 0& a_2 &
      \cdots & \,\,a_2+b_2& 0 & \cdots & 0\\
    \vdots & \vdots& \vdots&  \ddots &\vdots &\vdots & \vdots& \vdots&
      \ddots &\vdots & \vdots&  \ddots &\vdots & \vdots&  \ddots &\vdots \\
    0 & 0& 0 &  \cdots& a_n & a_n & 0& 0 & \cdots & 0& 0 &
      \cdots & 0& a_n+b_n \,\,& \cdots & a_n\\
    \vdots & \vdots& \vdots&  \ddots &\vdots &\vdots & \vdots& \vdots&
      \ddots &\vdots & \vdots&  \ddots &\vdots & \vdots&  \ddots &\vdots \\
    0 & 0& 0 &  \cdots& a_n & a_n & 0& 0 & \cdots & 0& 0 &
      \cdots & 0& a_n & \cdots & \,\,a_n+b_n
    \end{array}\right)\,.
 $$
 }
 % ------------- %
The corresponding matrix $M_{\mathrm{even}}$ is constructed in the way
described in the previous section. It consists of a nontrivial $2n
\times 2n$ part ($U_{\mathrm{r}}$ with added -1's) and
$\mathrm{deg\,}\mathcal{X}_1 - 2$ copies of
the row $(a_1,0,\dots,0,a_1)$, $\mathrm{deg\,}\mathcal{X}_2 - 2$
copies of the row $(0,a_2,a_2,0\dots,0)$, etc., and, finally, its
last $\mathrm{deg\,}\mathcal{X}_{n+1}$ rows have all the entries
equal to zero, hence the total number of its rows is $2N$ as
required.

If all the $a_j$'s are nonzero, the condition
$\mathrm{rank}\,M_{\mathrm{even}} < 2n$ simplifies to
 % ------------- %
 $$
   \mathrm{rank}\,\left(\begin{array}{cccccccc}
    b_1 & -1 & 0 & 0 & \cdots & 0 & 0 & 0\\
    -1 & b_2 & 0 & 0 & \cdots & 0 & 0 & 0\\
    0 & 0 & b_2 & -1 & \cdots & 0 & 0 & 0\\
    0 & 0 & -1 & b_3 & \cdots & 0 & 0 & 0\\
    \vdots & \vdots & \vdots & \vdots & \ddots & \vdots & \vdots & \vdots\\
    0 & 0 & 0 & 0 & \cdots & b_n & 0 & 0\\
    0 & 0 & 0 & 0 & \cdots & 0 & b_n & -1\\
    0 & 0 & 0 & 0 & \cdots & 0 & -1 & b_1\\
    1 & 0 & 0 & 0 & \cdots & 0 & 0 & 1\\
    0 & 1 & 1 & 0 & \cdots & 0 & 0 & 0\\
    \vdots & \vdots & \vdots & \vdots & \ddots & \vdots & \vdots & \vdots\\
    0 & 0 & 0 & 0 & \cdots & 1 & 1 & 0
   \end{array}\right)< 2n \,.
 $$
 % ------------- %
It is easy to see that the assumptions of Theorem~\ref{veta1}
giving rise to eigenvalues corresponding $kl_0 = 2\pi m$ are
satisfied in the case $b_j = -1, \;\forall j\in\{1,\dots N\}$,
which corresponds to $\delta$-couplings. The counterpart case,
$b_j = 1$, corresponding to $\delta_\mathrm{s}'$ couplings leads
to the requirement
 % ------------- %
 $$
   \mathrm{rank}\,\left(\begin{array}{ccccccccc}
    1 & 1 & 0 & 0 & 0 &\cdots & 0 & 0 & 0\\
    0 & 1 & -1 & 0 & 0 &\cdots & 0 & 0 & 0\\
    0 & 0 & 1 & 1 & 0 &\cdots & 0 & 0 & 0\\
    0 & 0 & 0 & 1 & -1 &\cdots & 0 & 0 & 0\\
    \vdots & \vdots & \vdots & \vdots & \vdots & \ddots & \vdots & \vdots & \vdots\\
    0 & 0 & 0 & 0 & 0 &\cdots & 1 & -1 & 0\\
    0 & 0 & 0 & 0 & 0 &\cdots & 0 & 1 & 1\\
    -1 & 0 & 0 & 0 & 0 &\cdots & 0 & 0 & 1\\
   \end{array}\right)< 2n \,.
 $$
 % ------------- %
which is satisfied if and only if the number of the edges in the
loop is even.

In a similar way, one can prove that eigenvalues corresponding to
$kl_0 = (2m+1)\pi$ are present in the spectrum of a graph with
$\delta_\mathrm{s}'$-couplings on the loop, while for
$\delta$-couplings this is true provided the loop consists of an
even number of the edges.

If there are several halflines attached to the loop and all the
$b_j$'s are equal to $-1$ or $+1$, respectively, we obtain the
same results as before. One can easily check that for $U = a J + b
I$ all entries of the energy-dependent part
 % ------------- %
 $$
   (1-k) U_2 [(1-k) U_4 - (k+1) I]^{-1} U_3
 $$
 % ------------- %
of the effective coupling matrix $\tilde U(k)$ given by
(\ref{en-dep}) are equal, hence the matrix $\tilde U$ can be
written using multiples of the matrices $J$ and $I$ and the
coefficient $b$ is not energy dependent, i.e. $\tilde U = \tilde
a(k) J + \tilde b I$. Since the coefficients $a_j$ can be
eliminated from the final condition, we obtain the same results as
in the energy-independent case.

Notice that the case $b_j = -1$ also includes an array of edges
with rationally related lengths and Dirichlet condition at the
both array endpoints. In this case one of the matrices describing
the coupling is $U_j = \mathrm{diag}\,(-1,-1)$. Similarly, $b_j =
1$ includes the case of an edge array with Neumann conditions at
both the endpoints, the corresponding matrix being $U_j =
\mathrm{diag}\,(1,1)$.

%%%%
%%%% Examples
%%%%
\section{Examples}

As stated in the introduction our main goal is to analyze
resonances which arise from the above discussed embedded
eigenvalues if the rational relation between the graph edge
lengths is perturbed. Let us look now at this effect in two
simple examples.

\subsection{A loop with two leads}

\begin{figure}
   \begin{center}
     \includegraphics{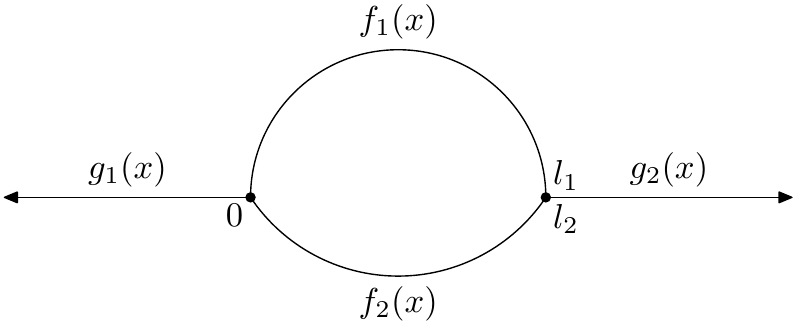}
     \caption{A loop with two leads}\label{figloop}
   \end{center}
\end{figure}

Consider first the graph sketched in Fig.~\ref{figloop} consisting
of two internal edges of lengths~$l_1,\,l_2$ and one halfline
connected at each endpoint. The Hamiltonian acts as
$-\mathrm{d}^2/\mathrm{d} x^2$ on each link. The corresponding
Hilbert space is $L^2(\mathbb{R}^+)\oplus L^2(\mathbb{R}^+)\oplus
L^2([0,l_1]) \oplus L^2([0,l_2])$; states of the system are
described by columns $\psi=(g_1,g_2,f_1,f_2)^T$. For a greater
generality, let us consider the following coupling conditions
\cite{ES} which include the $\delta$-coupling but allow also the
attachment of the semiinfinite links to the loop to be tuned, and
possibly to be turned off:
 % ------------- %
 \begin{eqnarray*}
    f_1(0)&=&f_2(0)\,,\qquad f_1(l_1)=f_2(l_2)\,,\\
    f_1(0) &=& \alpha_1^{-1} (f_1'(0)+f_2'(0))  + \gamma_1  g_1'(0)\,,\\
    f_1(l_1) &=& -\alpha_2^{-1} (f_1'(l_1)+f_2'(l_2))  + \gamma_2  g_2'(0)\,,\\
    g_1(0) &=& \bar\gamma_1 (f_1'(0)+f_2'(0))  + \tilde\alpha_1^{-1}  g_1'(0)\,,\\
    g_2(0) &=& -\bar\gamma_2 (f_1'(l_1)+f_2'(l_2))  + \tilde\alpha_2^{-1}   g_2'(0)\,.
 \end{eqnarray*}
 % ------------- %
Following the construction described in Sec.~\ref{sec2} and
parametrizing the internal edges by $l_1 = l(1-\lambda)$, $l_2 =
l(1+\lambda)$, $\lambda \in [0,1]$ --- which effectively means
shifting one of the connections points around the loop as
$\lambda$ is changing --- one arrives at the final condition for
resonances in the form
 % ------------- %
 \begin{eqnarray}
   \lefteqn{\sin{kl(1-\lambda)}\sin{kl(1+\lambda)}
   -4k^2\beta_1^{-1}(k)\beta_2^{-1}(k)\sin^2{kl}}
   \nonumber \\[.5em]  &&
   +k[\beta_1^{-1}(k)+\beta_2^{-1}(k)]\sin{2kl}=0\,, \label{1-con}
 \end{eqnarray}
 % ------------- %
where $\beta_i^{-1}(k) := \alpha_{i}^{-1}+ \frac{ik |\gamma_i|^2
}{1-ik \tilde\alpha_i^{-1}}$.

We are interested how the solutions to the above condition change
with respect to change of the length parameter $\lambda \to
\lambda'=\lambda+\varepsilon$. It is easy to check that any
solution $k$ depends on $\varepsilon$ continuously, and therefore
for small $\varepsilon$ we can thus construct a perturbation
expansion. Let $k_0$ be solution of (\ref{1-con}) for $\lambda$
and $k$ solution for $\lambda'$; the difference $\kappa = k - k_0$
can be obtained using the Taylor expansion
 % ------------- %
 \begin{eqnarray}
   \hspace{-15mm}\kappa l[\sin(2k_0l)-\lambda\sin(2k_0 l\lambda)]
   -4\kappa l k_0^2\beta_1^{-1}(k_0)\beta_2^{-1}(k_0)\sin{2 k_0 l}-
\nonumber\\
   \hspace{-10mm}-4\kappa[2k_0\beta_1^{-1}(k_0)\beta_2^{-1}(k_0)
   +k_0^2(\beta_1^{-1}(k_0)\tilde\beta_2(k_0)+
    \tilde\beta_1(k_0)\beta_2^{-1}(k_0))]\sin^2{k_0 l}+
\nonumber\\
   \hspace{-10mm}+\kappa(\beta_1^{-1}(k_0)+\beta_2^{-1}(k_0)
   +\tilde\beta_1(k_0)k_0+\tilde\beta_2(k_0)k_0)
    \sin{2k_0 l}+2\kappa l k_0(\beta_1^{-1}(k_0)+
\nonumber\\
   \hspace{-10mm}+\beta_2^{-1}(k_0))\cos{2k_0 l}
   -\kappa l[\varepsilon\cos{k_0l\varepsilon}\sin{k_0 l (2\lambda+\varepsilon)}+
\nonumber\\
   +(2\lambda+\varepsilon)\cos{k_0 l(2\lambda
   +\varepsilon)}\sin{k_0 l\varepsilon}]+
    \mathcal{O}(\kappa^2) =\sin{k_0 l(2\lambda
    +\varepsilon)}\sin{k_0 l\varepsilon}\,,\label{1-loopk1}
 \end{eqnarray}
 % ------------- %
where $\tilde\beta_j(k_0) =i|\gamma_j|^2/(1-ik_0
\tilde\alpha_j^{-1})^2$. This equation can be used to determine
$\kappa$ in the leading order. Denoting the coefficient of
$\kappa$ by $f(k_0)$ and the \emph{rhs} of the above equation by
$g(\lambda,\varepsilon)$ we find that the error in such an
evaluation is
 % ------------- %
 $$
   \delta = \frac{\mathcal{O}(\kappa^2)}{f(k_0)}
   =\frac{1}{f(k_0)}
  \,\mathcal{O}\left(\frac{g^2(\lambda,\varepsilon)}{f^2(k_0,
  \varepsilon)}\right)\,.
 $$
 % ------------- %
Since the \emph{rhs} of (\ref{1-loopk1}) is
$\mathcal{O}(\varepsilon)$ as $\varepsilon\to 0$, the error we
make by neglecting the term $\mathcal{O}(\kappa^2)$ is
$\mathcal{O}(\varepsilon^2)$. In fact, in the vicinity of the
embedded eigenvalues, i.e. for $2\lambda k_0 l$ close to $= 2n\pi$
the error is even smaller, namely $\mathcal{O}(\varepsilon^4)$ as
we will see below.

In fact, we can get more from eq.~(\ref{1-loopk1}) than just the
perturbative expansion. We are interested in the global behaviour,
i.e. trajectories of the resonance poles in the lower complex
halfplane as $\lambda$ changes. To obtained them one should solve
eq.~(\ref{1-con}), numerically since an analytic solution is
available in exceptional cases only. One can, however, solve also
numerically the approximate equation (\ref{1-loopk1}) starting
from $\lambda = \frac{m}{n}$ where corresponding the embedded
eigenvalues given by $kl = n\pi$ are present, and taking
$\varepsilon$ perturbations of the successive solutions. This
method is simple and we have employed it in the examples below,
with a sufficiently small step, $\varepsilon=5\cdot 10^{-5}$. To check
the consistency, we have compared the results in the second
example with a direct numerical solution of eq.~(\ref{1-con})
found with the step $0.05$ in the parameter $\lambda$, and found
that they give closely similar results, the relative error being
of order of $10^{-3}$.

Examples of poles trajectories obtained in the described way from
eq.~(\ref{1-loopk1}) are shown in
Figs.~\ref{figukazka1}--\ref{figkapka}. Eq.~(\ref{1-con}) has the
real solution $kl =n\pi,\,n\in \mathbb{N}$ for $\lambda = m/n,\,
m\in\mathbb{N}$, the corresponding eigenfunction is $\psi =
(0,0,\sin{n\pi x/l},-\sin{n\pi x/l})^T$. On Fig.~\ref{figukazka1}
corresponding to $n = 2$ the pole returns to the real axis when
$\lambda = 1/2$ and $\lambda = 1$. On the other hand,
Fig.~\ref{figpolovracec} with $n = 3$ shows the situation when the
pole returns to the real axis only for $\lambda = 2/3$, while for
$\lambda = 1/3$ and $\lambda = 1$ the appropriate solution is a
resonance. Similarly, the pole on Fig.~\ref{figkapka} where $n=2$
returns to the real axis only if $\lambda = 1$. To show how fast
the poles are moving, the change of the parameter $\lambda$ from 0
to 1 is marked by changing the colour from red ($\lambda = 0$) to
blue ($\lambda = 1$; visible online).

\begin{figure}
   \begin{center}
     \includegraphics{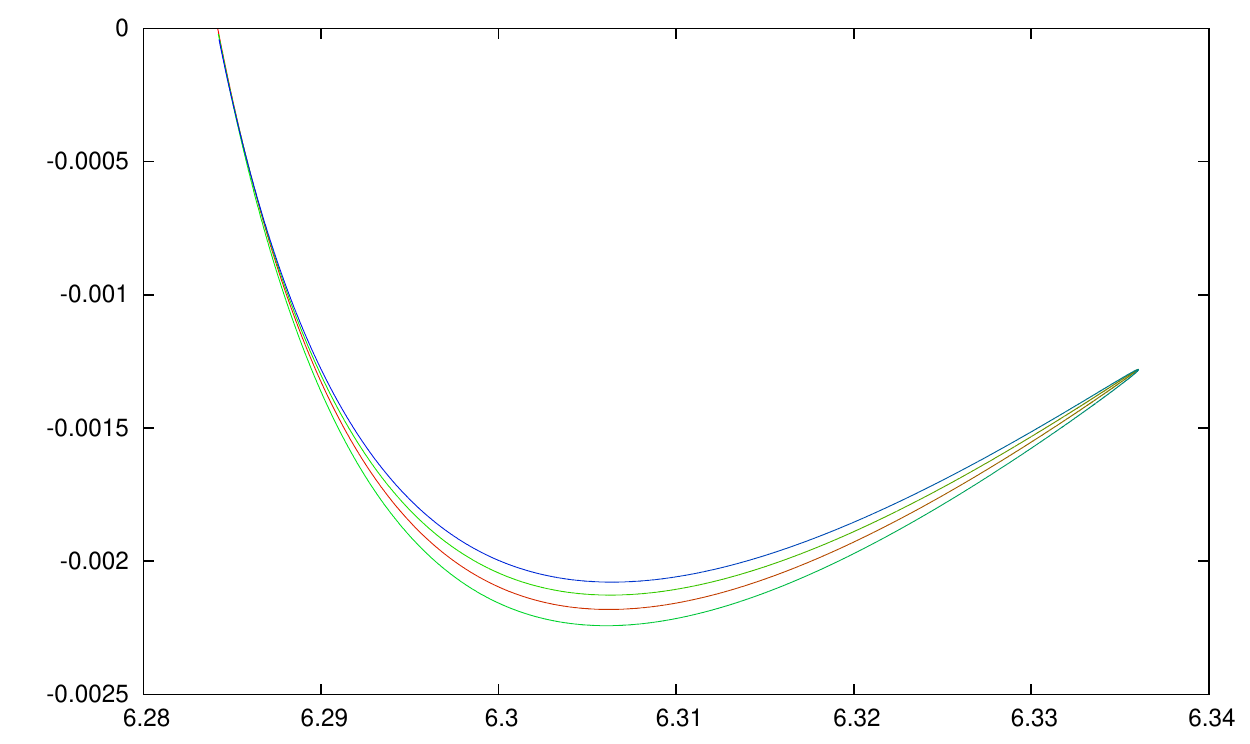}
     \caption{The trajectory of the resonance pole in the lower complex
     halfplane starting from $k_0 = 2\pi$ corresponding to $\lambda=0$ for $l=1$ and the coefficients
     values $\alpha_1^{-1}=1$, $\tilde\alpha_1^{-1}=-2$, $|\gamma_1|^2=1$,
     $\alpha_2^{-1}=0$, $\tilde\alpha_2^{-1}=1$, $|\gamma_2|^2=1$, $n=2$.
     The colour coding (visible online) shows the dependence on $\lambda$
     changing from red ($\lambda = 0$) to blue ($\lambda = 1$).}
     \label{figukazka1}
   \end{center}
\end{figure}

\begin{figure}
   \begin{center}
     \includegraphics{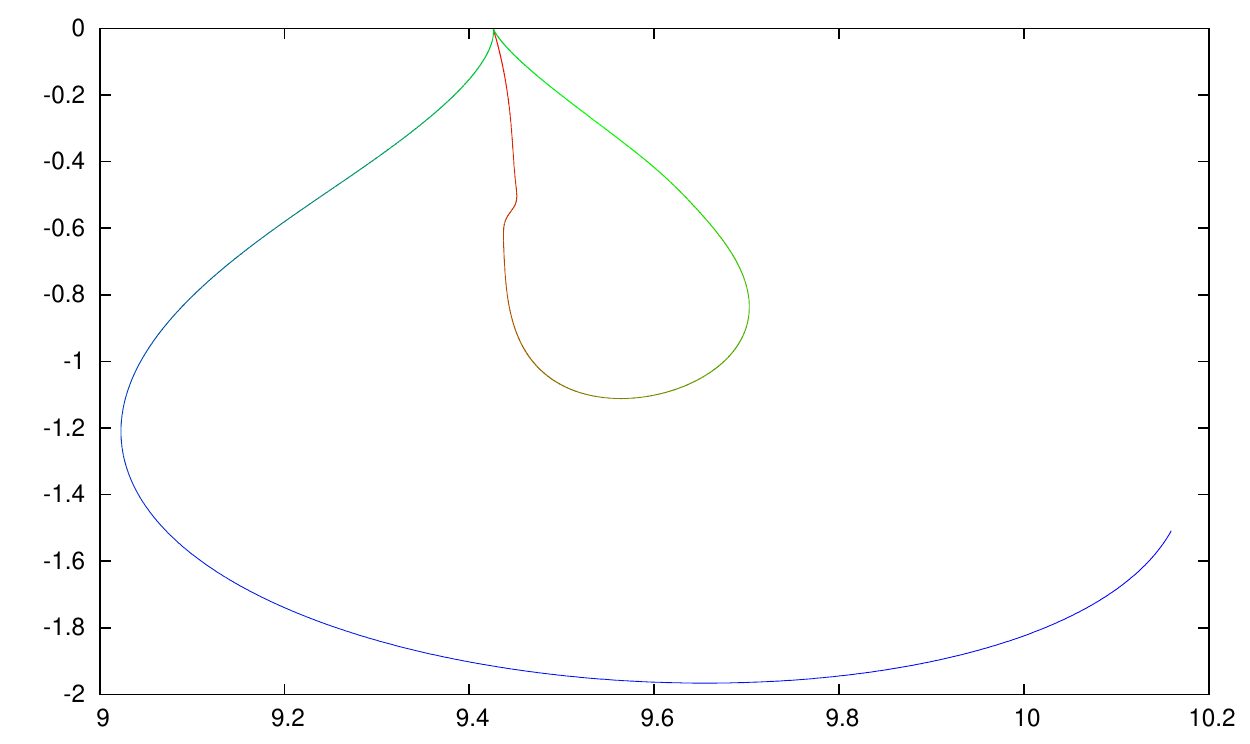}
     \caption{The trajectory of the resonance pole starting at $k_0 = 3\pi$ corresponding to $\lambda=0$ 
     for the coefficients values $\alpha_1^{-1}=1$, $\alpha_2^{-1}=1$,
     $\tilde\alpha_1^{-1}=1$, $\tilde\alpha_2^{-1}=1$,
     $|\gamma_1|^2=|\gamma_2|^2=1$, $n=3$. The colour coding (visible
     online) is the same as in the previous picture.}
     \label{figpolovracec}
   \end{center}
\end{figure}

\begin{figure}
   \begin{center}
     \includegraphics{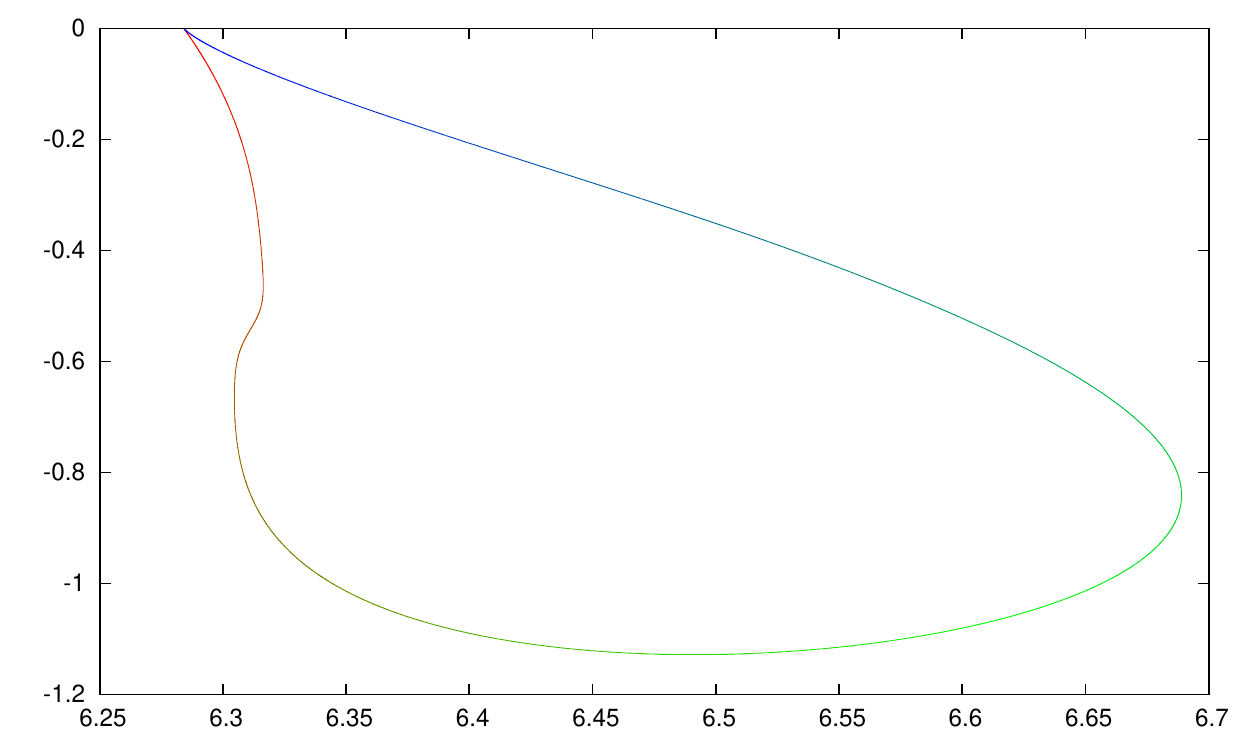}
     \caption{The trajectory of the resonance pole starting at $k_0 = 2\pi$ corresponding to $\lambda=0$ for
      the coefficients values $\alpha_1^{-1}=1$, $\alpha_2^{-1}=1$,
      $\tilde\alpha_1^{-1}=1$, $\tilde\alpha_2^{-1}=1$, $|\gamma_1|^2=1$,
      $|\gamma_2|^2=1$, $n=2$. The colour coding is the same as above. }
     \label{figkapka}
   \end{center}
\end{figure}

Let us now investigate the asymptotic behaviour of the resonances
in the vicinity of the embedded eigenvalue, in particular, the
angle~$\varphi$ between the pole trajectory emerging from $k_0=
n\pi/l$ with $\lambda_0 = m/n, m\in\{0,1,\dots,n\}$ and the real
axis. For small $\kappa$ the difference $\varepsilon =
\lambda-\lambda_0$ is also small. We use a rewritten form of the
condition (\ref{1-con}),
 % ------------- %
 \begin{eqnarray}
   f(k,\lambda) &\!=\!& \cos{2kl\lambda}-\cos{2kl}
   -8k^2\beta_1^{-1}(k)\beta_2^{-1}(k)\sin^2{kl} \nonumber
   \\ && +2k(\beta_1^{-1}(k)+\beta_2^{-1}(k))\sin{2kl} = 0\,.
   \label{1-con2}
 \end{eqnarray}
 % ------------- %
The function $f(k,\lambda)$ is, with the exception of points $k =
-i\tilde\alpha_j$, continuous and its first partial derivative
with respect to $\lambda$ is at $\lambda_0$ is equal to zero,
hence
 % ------------- %
 \begin{eqnarray*}
   \phantom{AAAAAA}0 &=& f(k,\lambda)\approx f(k_0,\lambda_0)
   +\left.\frac{\partial^2 f}{\partial \lambda^2}\right|_{k_0, \lambda_0}
    \varepsilon^2+\left.\frac{\partial f}{\partial k}\right|_{k_0, \lambda_0}
    \kappa\,,
 \\
   \left.\frac{\partial f}{\partial k}\right|_{(k_0, \lambda_0)}
   &=& 4n\pi\left[\beta_1^{-1}(k_0)
    +\beta_2^{-1}(k_0)\right]\,,
 \\
    \left.\frac{\partial^2 f}{\partial \lambda^2}\right|_{(k_0, \lambda_0)} &=&
    -4(kl)^2\cos{2kl\lambda}=-4(\pi n)^2\,.
 \end{eqnarray*}
 % ------------- %
For small $\kappa$ we obtain using (\ref{1-loopk1})
 % ------------- %
 \begin{eqnarray}
    \phantom{AAAAAAA}\kappa &\approx& \varepsilon^2
    \frac{\pi n}{\beta_1^{-1}(k_0)+\beta_2^{-1}(k_0)}\,,\nonumber
\\ [.5em]
    \tan{\varphi} =\frac{\mathrm{Im}\,\kappa}{\mathrm{Re}\,\kappa}
    &=& \frac{\frac{k_0|\gamma_1|^2}{1+k_0^2\tilde\alpha_1^{-2}}
    +\frac{k_0|\gamma_2|^2}{1+k_0^2\tilde\alpha_2^{-2}}}{\alpha_1^{-1}
    +\alpha_2^{-1}-\frac{k_0^2|\gamma_1|^2
    \tilde\alpha_1^{-1}}{1+k_0^2\tilde\alpha_1^{-2}}-\frac{k_0^2|\gamma_2|^2
    \tilde\alpha_2^{-1}}{1+k_0^2\tilde\alpha_2^{-2}}}\,,\quad k_0
    = \frac{n\pi}{l}\,.\phantom{xxxx}\label{1-angle1}
 \end{eqnarray}
 % ------------- %
For $|\gamma_1| = |\gamma_2| = 0$ the poles are real and $\varphi
=0$; this is the case when the loop and the leads are decopupled
and the eigenvalues remain embedded. On the other hand, if
$\alpha_1^{-1} = \tilde\alpha_1^{-1} = \alpha_2^{-1} =
\tilde\alpha_2^{-1} = 0$ then the real part of $\kappa$ is zero
and the pole trajectory goes from $k_0$ perpendicular to the
horizontal line, i.e. $\varphi  = \pi /2$.

Furthermore, let us investigate the behavior of the pole
trajectories hight in the spectrum, i.e. for large values of $n$.
Suppose that $k = k_0 + \kappa$, $k_0 = n\pi/l$, $|\kappa| \ll
\pi/l$; then
 % ------------- %
 \begin{eqnarray*}
   \cos{2kl\lambda}-\cos{2kl} &=& \cos{2k_0l\lambda}\cos{2\kappa l\lambda}
   - \sin{2k_0l\lambda}\sin{2\kappa l\lambda} -\cos{2\kappa l} \\
   &=& (\cos{2n\pi\lambda}- 1)- \sin{(2\pi n\lambda)} \,2\kappa l\lambda
   + {\mathcal O}(\kappa^2)\,.
 \end{eqnarray*}
 % ------------- %
The condition (\ref{1-con2}) for small $\kappa$ becomes
 % ------------- %
 $$
   (\cos{2n\pi\lambda}-1)-\sin{(2\pi n\lambda)} \,2\kappa l\lambda
    +2\frac{n\pi}{l}\left[\beta_1^{-1}(k_0)
    +\beta_2^{-1}(k_0)\right]2\kappa l + {\mathcal O}(\kappa^2) =0\,.
 $$
 % ------------- %
Using the expressions of coefficients $\beta_j (k)$ we obtain
 % ------------- %
 \begin{eqnarray*}
   \beta_j^{-1}(k_0) = \alpha_j^{-1}-\frac{|\gamma_j|^2}{\tilde\alpha_j^{-1}}
    +i\frac{l|\gamma_j|^2}{n\pi \tilde\alpha_j^{-2}}+{\mathcal O}(n^{-2})
    \qquad \hbox{for}\quad \tilde\alpha_j^{-1}\not = 0\,, \\
   \beta_j^{-1}(k_0) = i\frac{n\pi}{l}|\gamma_j|^2 + {\mathcal O}(1)
    \qquad \hbox{for}\quad \tilde\alpha_j^{-1} = 0\,.
 \end{eqnarray*}
 % ------------- %
The quantities appearing above,
 % ------------- %
 $$
  |\cos{(2n\pi\lambda)}-1| \leq 2 \quad \textrm{and}\quad
  |\sin{(2\pi \kappa l\lambda)}|\leq 1\,
 $$
 % ------------- %
are bounded, thus for $\tilde\alpha_1^{-1}\not = 0$ and
$\tilde\alpha_2^{-1}\not = 0$, we have
 % ------------- %
 \begin{equation*}
   |\mathrm{Im} \, \kappa| \leq \frac{l}{2(\pi n)^2}
    \frac{|\gamma_1|^2/\tilde\alpha_1^{-2}+|\gamma_2|^2/\tilde\alpha_2^{-2}}{(\alpha_1^{-1}+
    \alpha_2^{-1}-|\gamma_1|^2/\tilde\alpha_1^{-1}-|\gamma_2|^2/\tilde\alpha_2^{-1})^2}+
    {\mathcal O}(n^{-3})\,,
 \end{equation*}
 % ------------- %
while for $\tilde\alpha_1^{-1}= 0$ and $\tilde\alpha_2^{-1} = 0$
the inequality reads
 % ------------- %
 \begin{equation*}
   |\mathrm{Im} \, \kappa| \leq \frac{l}{2(\pi n)^2}
    \frac{1}{|\gamma_1|^2+|\gamma_2|^2}+{\mathcal O}(n^{-3})\,,
 \end{equation*}
 % ------------- %
and for $\tilde\alpha_1^{-1} = 0,\,\tilde\alpha_2^{-1}\not = 0$ we have
 % ------------- %
 \begin{equation*}
   |\mathrm{Im} \, \kappa| \leq \frac{l}{2(\pi n)^2}
    \frac{1}{|\gamma_1|^2}+{\mathcal O}(n^{-3})\,.
 \end{equation*}
 % ------------- %

Let us summarize the discussion of this example. The poles of the
resolvent are given by the condition (\ref{1-con}), or
equivalently, by (\ref{1-con2}). If $\lambda = m/n,\,
m\in\mathbb{N}$, real eigenvalues corresponding to $kl
=n\pi,\,n\in \mathbb{N}$, occur. They may correspond to a
particular pole of the resolvent returning to the real axis for
$\lambda = m/n,\, m\in\mathbb{N}$, as in Fig.~\ref{figukazka1}.
However, for other coupling conditions, the pole may return only
for certain $\lambda$ --- see Figs.~\ref{figpolovracec}
and~\ref{figkapka}, while for other rational $\lambda$ its place
may be taken by the pole which has been a resonance for $\lambda =
0$. The angle between the resonance trajectory and the real axis
does not depend on $\lambda$ and is given by (\ref{1-angle1}). If
the pole trajectory is near the original eigenvalue, then the
distance from the real axis is of order of ${\mathcal O}(n^{-2})$
for large $n$.

\subsection{A cross-shaped graph}

Let us now consider another simple graph, this time
consisting of two leads and two internal edges attached to the
leads at one point -- cf.~Fig.~\ref{figresonator}; the lengths of
the internal edges are $l_1 = l(1-\lambda)$ and
$l_2=l(1+\lambda)$.
 % ------------- %
\begin{figure}
   \begin{center}
     \includegraphics{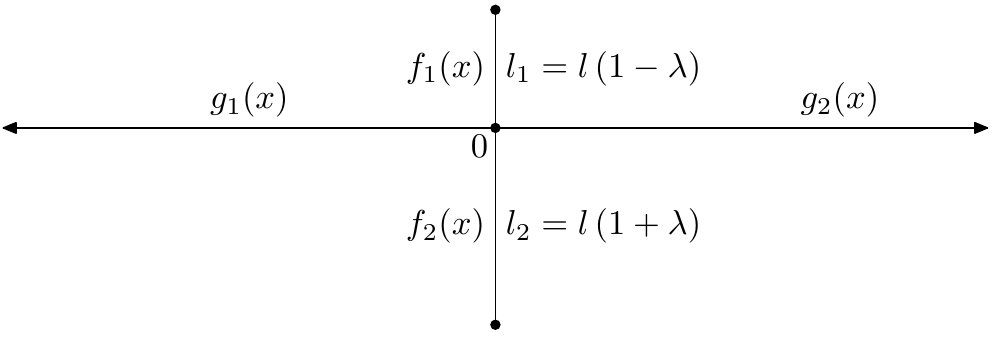}
     \caption{A cross-shaped resonator}
     \label{figresonator}
   \end{center}
\end{figure}
 % ------------- %
The Hamiltonian acts again as $-\mathrm{d}^2/\mathrm{d}x^2$ on the
corresponding Hilbert space $L^2(\mathbb{R}^+)\oplus
L^2(\mathbb{R}^+)\oplus L^2([0,l_1]) \oplus L^2([0,l_2])$, and the
states are described by columns $\psi=(g_1,g_2,f_1,f_2)^T$. This
time we restrict ourselves to the $\delta$ coupling as the
boundary condition at the vertex and we consider Dirichlet
conditions at the loose ends, i.e.
 % ------------- %
 \begin{eqnarray*}
    f_1(0)&=&f_2(0)=g_1(0)=g_2(0)\,,\\
    f_1(l_1)&=&f_2(l_2)=0\,,\\
    \alpha f_1(0) &=& f_1'(0)+f_2'(0)+ g_1'(0)+ g_2'(0)\,.
 \end{eqnarray*}
 % ------------- %
Using the same technique as above we arrive at two equivalent
forms of the condition for resonances, $k\sin{2kl}
+(\alpha-2ik)\sin{kl(1-\lambda)} \sin{kl(1+\lambda)}=0$ or
 % ------------- %
 \begin{equation}
    2k\sin{2kl}+(\alpha-2ik)
    (\cos{2kl\lambda-\cos{2kl}})=0\,.\label{1-condresonator}
 \end{equation}
 % ------------- %
Let us ask when the solution is real. Leaving out the trivial case
$k = 0$ we get from the last equation two conditions referring to the
vanishing of the real and imaginary parts of the \emph{lhs},
 % ------------- %
 \begin{eqnarray*}
    \hspace{4.2em}\sin{2kl}=0\quad &\Rightarrow&
    \quad kl=\frac{n\pi}{2},\quad n\in {\mathbb Z}\,,
    \\
    \hspace{-.5em}0=\cos{2kl\lambda}-\cos{2kl} &=&\cos{n\pi\lambda}-\cos{n\pi}=
2\sin{\frac{n\pi}{2} (1-\lambda)}\sin{\frac{n\pi}{2}(1+\lambda)}
    \\
\hspace{-.5em}&\Rightarrow&  \quad n\lambda=(n-2m),\quad
m\in{\mathbb Z}\,.
 \end{eqnarray*}
 % ------------- %
Hence $\lambda = 1-2m/n$, $\,m\in \mathbb{N}_0$, $\,m\leq n/2$.
 % ------------- %
\begin{figure}
   \begin{center}
     \includegraphics{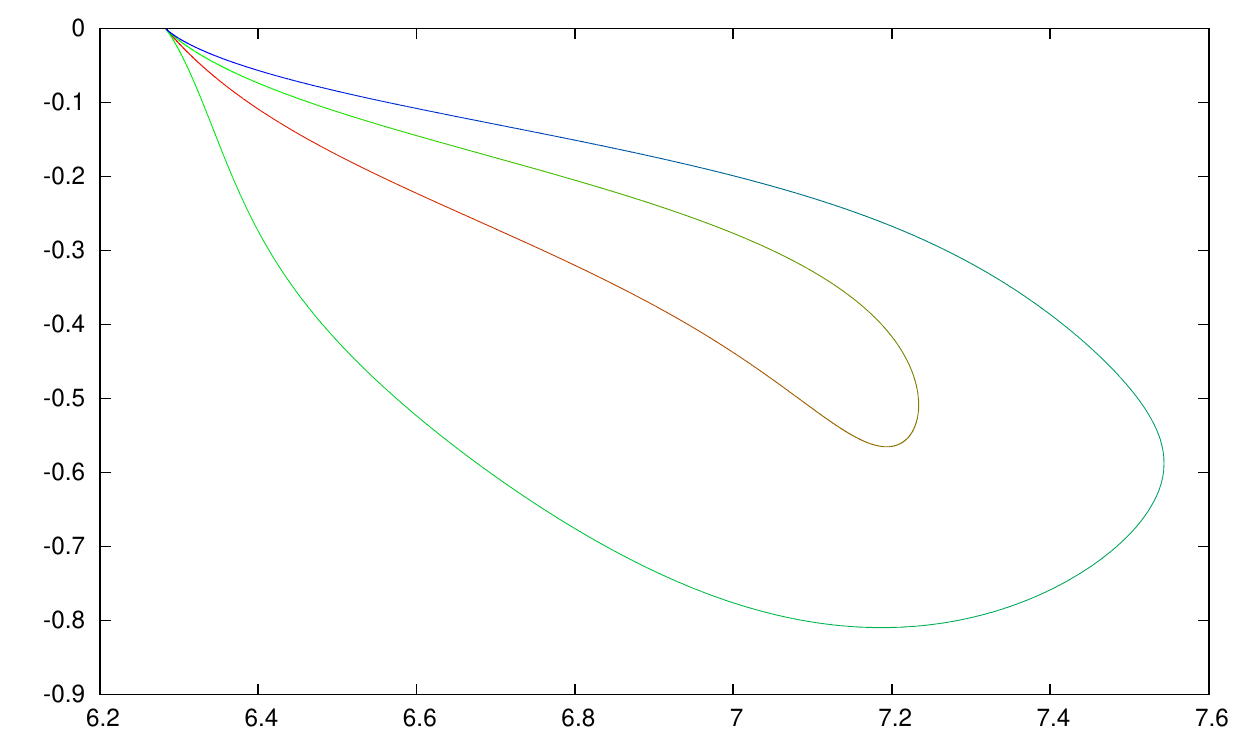}
     \caption{
The trajectory of the resonance pole starting at $k_0 = 2\pi$ for
the coefficients values $\alpha = 10$, $n=2$. The colour coding
(visible online) is the same as in the previous figures.}
     \label{figrez10}
   \end{center}
\end{figure}
 % ------------- %
\begin{figure}
   \begin{center}
     \includegraphics{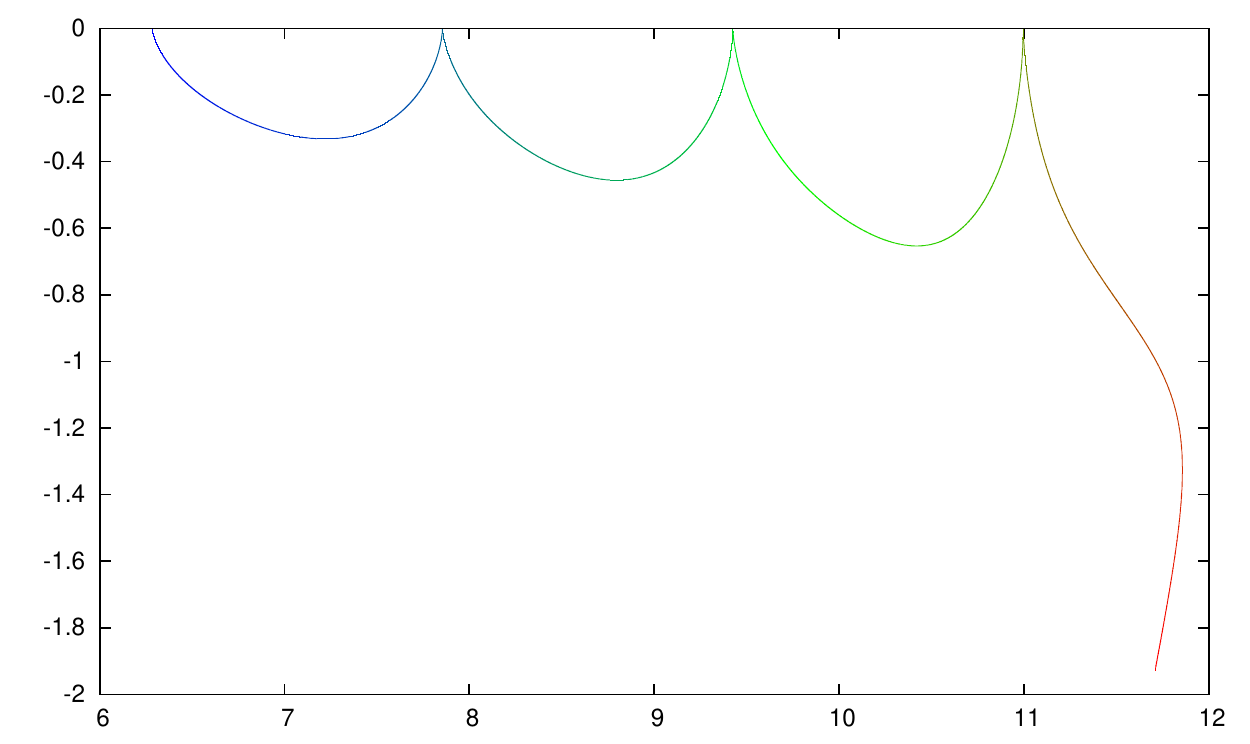}
     \caption{The trajectory of the resonance pole for the
     coefficients values $\alpha = 1$, $n=2$. The colour coding is
     the same as above.}
     \label{figrez1b}
   \end{center}
\end{figure}
 % ------------- %
\begin{figure}
   \begin{center}
     \includegraphics{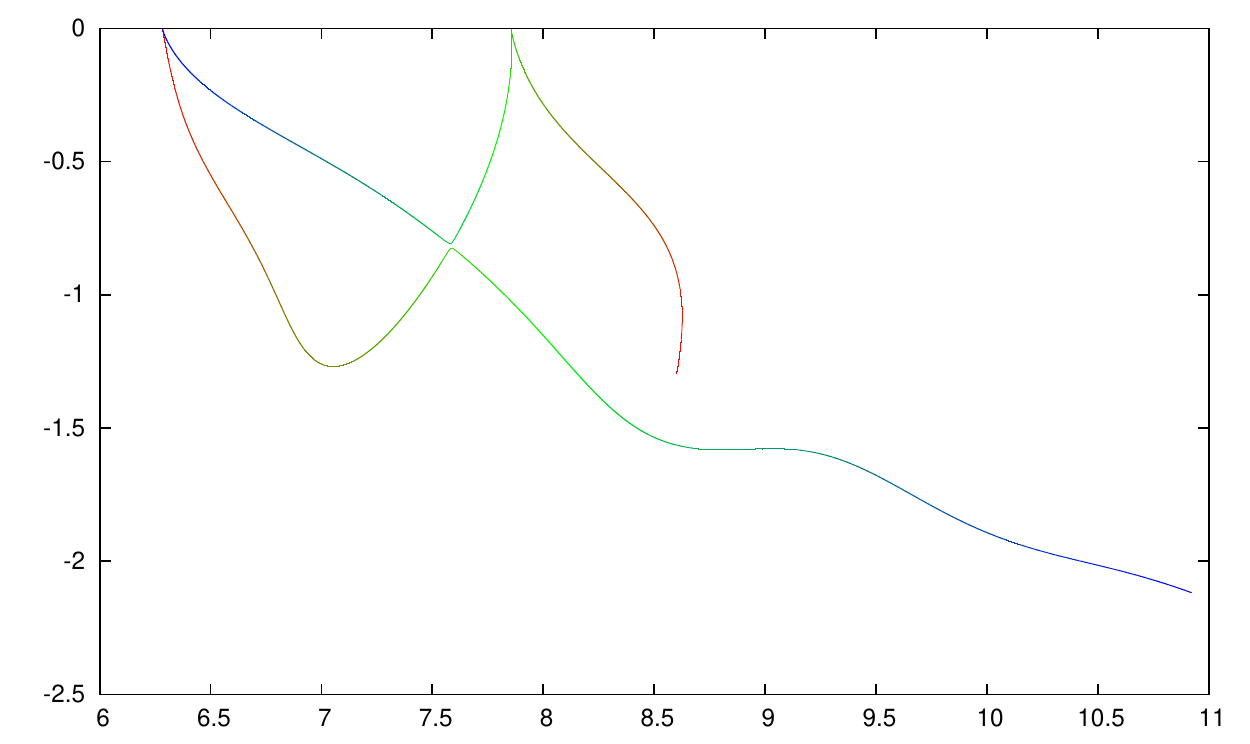}
     \caption{The trajectories of two resonance poles for the
     coefficients values $\alpha = 2{.}596$, $n=2$. We can see
     an avoided resonance crossing -- the former eigenvalue ``travelling from
     the left to the right'' interchanges with the former resonance
     ``travelling the other way'' and ending up as an embedded
     eigenvalue. The colour coding is the same as above.}
     \label{figrez1c}
   \end{center}
\end{figure}
 % ------------- %
If the difference $\kappa=k-k_0$ is small we obtain from
(\ref{1-condresonator})
 % ------------- %
 \begin{eqnarray}
    \hspace{-5em}\lefteqn{\kappa \approx -2(\alpha-2ik_0)\sin{k_0l\varepsilon}
    \sin{k_0l(2\lambda+\varepsilon)} \Big\{2i[\cos{2k_0l
    (\lambda+\varepsilon)}-\cos{2k_0l}]}
    \nonumber\\[.3em]
    && \hspace{-4.5em}+(\alpha-2ik_0)2l[(\lambda+\varepsilon)\sin{2k_0l(\lambda+\varepsilon)}
    -\sin{2k_0l}]-2\sin{2k_0l}-4k_0l\cos{2k_0l}\Big\}^{-1}.\label{1-vypocet}
 \end{eqnarray}
 % ------------- %
Similarly as in the previous example, the error here is
$\mathcal{O}(\kappa^2)$, i.e. $\mathcal{O}(\varepsilon^2)$, and
for $2\lambda k_0 l$ close to $= 2n\pi$ it is even smaller, namely
$\mathcal{O}(\varepsilon^4)$. In the latter case, the above
expression for $k_0 = n\pi/l$, $\lambda = m/n$ and small
$\varepsilon$ yields
 % ------------- %
 $$
    \kappa \approx -\frac{2(\alpha-2ik_0)(k_0l\varepsilon)^2}{-4k_0l}
    =\frac{n\pi\varepsilon^2}{2}
    \left(\alpha-2i\frac{n\pi}{l}\right)\,.
 $$
 % ------------- %
The slope of the pole trajectory at its start from $k_0$  is equal
to
 % ------------- %
 \begin{equation}
    \tan{\varphi}=-\frac{\mathrm{Im}\,\kappa}{\mathrm{Re}\,\kappa}
    =\frac{2n\pi}{\alpha l}\quad\Rightarrow \quad\varphi
    = \arctan{\frac{2n\pi}{\alpha l}}\,.\label{1-angle2}
 \end{equation}
 % ------------- %

As we have said, the embedded eigenvalues occur in accordance with
(\ref{1-condresonator}) at $kl=n\pi/2$, $n\in \mathbb{Z}$ for
$\lambda = 1-2m/n$, $m\in \mathbb{N}_0$, $m\leq n/2$. The
geometric perturbation gives rise to pole trajectories which can
be found from (\ref{1-condresonator}), or from (\ref{1-vypocet})
with a sufficiently small step. Examples worked out using the
second method are on
Figs.~\ref{figrez10}--\ref{figrez1c}. We see that a resolvent pole
may return to the same point, or it may become another eigenvalue
or a resonance. Another interesting type of behaviour, an avoided
resonance crossing, can be seen on Fig.~\ref{figrez1c}.

%%%%
%%%% The general case
%%%%
\section{The general case}

After analyzing the above two examples, let us look what could be
said about the geometric perturbation problem in the general case.

\subsection{Multiplicity of the eigenvalues}

Suppose that $k_0$ is an eigenvalue of multiplicity $d$ embedded
in a continuous spectrum of $H$. First we will assume that $k_0
l_0 = 2\pi m$. Our aim is now to determine whether $k_0$ is still
eigenvalue (and what is its multiplicity) if the lengths of the
graph edges are perturbed. We will write the lengths as $l_j' =
l_0 (n_j + \varepsilon_j)$ assuming that $n_j \in \mathbb{N}$ for
$j \in \{1,\dots , n\}$, while $n_j$ is not an integer for $j \in
\{n+1,\dots , N\}$.

From the construction described in the proof of theorem
\ref{veta1} we find that the condition (\ref{rat-matice}) is not
affected by small lengths variations of the ``nointeger'' edges,
$j\in \{n+1,\dots , N\}$. Hence the number of rationality-related
eigenvalues of the perturbed graph referring to the first $n$
edges does not depend on perturbations of the other edge
lengths. The spectral condition (\ref{rat-con-l2}) can be written
as $\mathrm{det}\,J(k)=0$ if we put $J(k) := C(k) + S(k)$. Using
the expansion
 % ------------- %
 \begin{eqnarray*}
ik \cos{\frac{kl_0(n_j+\varepsilon_j)}{2}}\mp
\sin{\frac{kl_0(n_j+\varepsilon_j)}{2}}
 \\
  \quad =\cos{\frac{k_0 l_0n_j}{2}} \left(ik_0 \cos{\frac{k_0 l_0
  \varepsilon_j}{2}}\mp \sin{\frac{k_0 l_0 \varepsilon_j}{2}}\right) +
  \mathcal{O}(k-k_0)\,,
 \end{eqnarray*}
 % ------------- %
and an analogous one for $\cos{\frac{kl_0(n_j+\varepsilon_j)}{2}}+
i k \sin{\frac{kl_0(n_j+\varepsilon_j)}{2}}$ one finds that the
entries of $J(k)$ can be rewritten as
 % ------------- %
 \begin{eqnarray*}
 \hspace{-10mm} J_{i,2j-1}(k) =(u_{i,2j-1} - u_{i,2j})\cos{\frac{k_0 l_0n_j}{2}}
\left(ik_0 \cos{\frac{k_0 l_0 \varepsilon_j}{2}}
- \sin{\frac{k_0 l_0 \varepsilon_j}{2}}\right)+
\\
(\delta_{i,2j-1} -\delta_{i,2j})
\cos{\frac{k_0 l_0n_j}{2}} \left(ik_0 \cos{\frac{k_0 l_0 \varepsilon_j}{2}}
+ \sin{\frac{k_0 l_0 \varepsilon_j}{2}}\right)+ \mathcal{O}(k-k_0)
\\
 \hspace{-10mm} J_{i,2j}(k) = (u_{i,2j-1} + u_{i,2j})\cos{\frac{k_0 l_0n_j}{2}}
\left(\cos{\frac{k_0 l_0 \varepsilon_j}{2}}
+ i k_0 \sin{\frac{k_0 l_0 \varepsilon_j}{2}}\right)+
\\
(\delta_{i,2j-1} +\delta_{i,2j})
\cos{\frac{k_0 l_0n_j}{2}} \left(-\cos{\frac{k_0 l_0 \varepsilon_j}{2}}
+ ik_0\sin{\frac{k_0 l_0 \varepsilon_j}{2}}\right)+ \mathcal{O}(k-k_0)
 \end{eqnarray*}
 % ------------- %
For small enough $\varepsilon_j$'s and a real nonzero noninteger $k_0$ the terms
$\cos{\frac{k_0 l_0n_j}{2}}$, $\,i k_0\cos{\frac{k_0 l_0
\varepsilon_j}{2}} -\sin{\frac{k_0 l_0 \varepsilon_j}{2}}$ and
$\cos{\frac{k_0 l_0 \varepsilon_j}{2}}+ i k_0 \sin{\frac{k_0 l_0
\varepsilon_j}{2}}$ are nonzero. After dividing the columns of
$J(k)$ by these terms and using the arguments from the proof of
Theorem~\ref{veta1} one arrives at the following conclusion.
 % ------------- %
 \begin{theorem}
 In the setting of Theorem~\ref{veta1} suppose that the rank of $M_\mathrm{even}$
 is smaller than $2n$. Let us vary the edge lengths,
 $l_j' = l_0 (n_j + \varepsilon_j)$ with sufficiently small $\varepsilon_j$'s; then the
 multiplicity of the eigenvalues $\epsilon = k_0^2 = 4m^2\pi^2 /l_0^2$ due to
 rationality of the first $n$ edges is given by the difference between $2n$
 and the rank of the matrix
{\scriptsize
% ------------- %
   \begin{equation*}
   \hspace{-25mm}M_{\mathrm{even}}^{\{\varepsilon_j\}} = \left(\begin{array}{ccccccc}
   u_{11}+\tilde\varepsilon_{1}^a& u_{12}-1+\tilde\varepsilon_{1}^b&u_{13}&u_{14}&\cdots&u_{1,2n-1}&u_{1,2n}\\
   u_{21}-1+\tilde\varepsilon_{1}^b& u_{22}+\tilde\varepsilon_{1}^a&u_{23}&u_{24}&\cdots&u_{2,2n-1}&u_{2,2n}\\
   u_{31}& u_{32}&u_{33}+\tilde\varepsilon_{2}^a&u_{34}-1+\tilde\varepsilon_{2}^b&\cdots&u_{3,2n-1}&u_{3,2n}\\
   u_{41}& u_{42}&u_{43}-1+\tilde\varepsilon_{2}^b&u_{44}+\tilde\varepsilon_{2}^a&\cdots&u_{4,2n-1}&u_{4,2n}\\
   \vdots&\vdots&\vdots&\vdots&\ddots&\vdots&\vdots\\
   u_{2N-1,1}& u_{2N-1,2}&u_{2N-1,3}&u_{2N-1,4}&\cdots&u_{2N-1,2n-1}&u_{2N-1,2n}\\
   u_{2N,1}& u_{2N,2}&u_{2N,3}&u_{2N,4}&\cdots&u_{2N,2n-1}&u_{2N,2n}
   \end{array}\right),
 \end{equation*}
 % ------------- %
} \noindent
where
 % ------------- %
 {\scriptsize
 $$
   \tilde \varepsilon_{j}^a (k):= \frac{(1-k_0^2)\sin{k_0 l_0 \varepsilon_j}}
   {2ik_0 \cos{k_0 l_0 \varepsilon_j} -(1+k_0^2)\sin{k_0 l_0 \varepsilon_j}}\,,\quad
    \tilde \varepsilon_{j}^b (k):= \frac{2ik_0(-1+\cos{k_0 l_0 \varepsilon_j})-(1+k_0^2)\sin{k_0 l_0 \varepsilon_j}}
   {2ik_0 \cos{k_0 l_0 \varepsilon_j} -(1+k_0^2)\sin{k_0 l_0 \varepsilon_j}}\,.
 $$
 }
 % ------------- %
 \end{theorem}
 % ------------- %

In a similar way one can treat the case when $k_0l_0$ is equal to
odd multiples of $\pi$. Then we employ the expansion
 % ------------- %
 {\small
 \begin{eqnarray*}
\hspace{-25mm}ik \cos{\frac{kl_0(n_j+\varepsilon_j)}{2}}
\mp \sin{\frac{kl_0(n_j+\varepsilon_j)}{2}}
=  \sin{\frac{k_0 l_0n_j}{2}} \left(-ik_0 \sin{\frac{k_0 l_0 \varepsilon_j}{2}}
\mp \cos{\frac{k_0 l_0 \varepsilon_j}{2}}\right) + \mathcal{O}(k-k_0)
 \end{eqnarray*}
 % ------------- %
and an analogous expression for
$\cos{\frac{kl_0(n_j+\varepsilon_j)}{2}} + i k
\sin{\frac{kl_0(n_j+\varepsilon_j)}{2}}$; with the help of them we
arrive at the following conclusion.

 % ------------- %
 \begin{theorem}
In the setting of Theorem~\ref{veta2} suppose that the rank of
$M_\mathrm{odd}$ is smaller than $2n$. Passing to $l_j' = l_0 (n_j
+ \varepsilon_j)$ with small enough $\varepsilon_j$'s, the
multiplicity of the eigenvalues $\epsilon = k_0^2 = (2m+1)^2\pi^2
/l_0^2$ due to rationality of the first $n$ edges is given by the
difference between $2n$ and rank of a matrix
{\scriptsize
% ------------- %
 \begin{equation*}
   \hspace{-25mm}M_{\mathrm{odd}}^{\{\varepsilon_j\}} = \left(\begin{array}{ccccccc}
   u_{11}+\tilde\varepsilon_{1}^a& u_{12}+1-\tilde\varepsilon_{1}^b&u_{13}&u_{14}&\cdots&u_{1,2n-1}&u_{1,2n}\\
   u_{21}+1-\tilde\varepsilon_{1}^b& u_{22}+\tilde\varepsilon_{1}^a&u_{23}&u_{24}&\cdots&u_{2,2n-1}&u_{2,2n}\\
   u_{31}& u_{32}&u_{33}+\tilde\varepsilon_{2}^a&u_{34}+1-\tilde\varepsilon_{2}^b&\cdots&u_{3,2n-1}&u_{3,2n}\\
   u_{41}& u_{42}&u_{43}+1-\tilde\varepsilon_{2}^b&u_{44}+\tilde\varepsilon_{2}^a&\cdots&u_{4,2n-1}&u_{4,2n}\\
   \vdots&\vdots&\vdots&\vdots&\ddots&\vdots&\vdots\\
   u_{2N-1,1}& u_{2N-1,2}&u_{2N-1,3}&u_{2N-1,4}&\cdots&u_{2N-1,2n-1}&u_{2N-1,2n}\\
   u_{2N,1}& u_{2N,2}&u_{2N,3}&u_{2N,4}&\cdots&u_{2N,2n-1}&u_{2N,2n}
   \end{array}\right)
 \end{equation*}
 % ------------- %
} \noindent
with $\tilde\varepsilon_{j}^a$ and $\tilde\varepsilon_{j}^b$ defined in previous theorem.
 \end{theorem}
 % ------------- %

\subsection{Total number of poles of the resolvent after perturbation}

In general an embedded eigenvalue can split under the geometric
perturbations considered here, a part of it being preserved with a
lower multiplicity while the rest is turned into resonance(s).
Above we have shown what the reduced multiplicity of the embedded
eigenvalue is, now we complement this result by showing that the
total number of poles produced in this way, multiplicity taken
into account, remains locally preserved. Before stating the
result, let us first demonstrate two useful lemmata.

 % ------------- %
 \begin{lemma}\label{rat-lemma1}
Let $(k,\vec \varepsilon)\mapsto g(k,\vec \varepsilon) :
\mathbb{C}\times \mathbb{R}^m \to \mathbb{C}$ be a function
uniformly continuous in $\vec \varepsilon$ for all
$\vec\varepsilon \in \mathcal{U}_{\varepsilon_0}(0)$ and $k \in
\mathcal{U}_R(k_0)$, $\varepsilon_0>0, \,R>0$, and holomorphic in
$k$ in $\mathcal{U}_R(k_0)$ for all $\vec\varepsilon \in
\mathcal{U}_{\varepsilon_0}(0)$. Furthermore, let
$\lim_{\vec\varepsilon\to 0} g(k,\vec\varepsilon) = (k - k_0)^d$.
Then there exist such $\delta>0$ and $\varepsilon_0'>0$ that for all
$\vec\varepsilon \in \mathcal{U}_{\varepsilon_0'}(0)$ the sum of
the multiplicities of zeros of $g(k,\vec\varepsilon)$ in
$\mathcal{U}_\delta (k_0)$ is $d$.
 \end{lemma}
 \begin{proof}
Since $g$ is holomorphic, we have the Taylor expansion
 % ------------- %
 $$
   g(k,\vec\varepsilon)
 = \sum_{p=0}^\infty a_p(\varepsilon)(k-k_0)^p
 = P(k,\vec\varepsilon)+ (k - k_0)^{d+1} h(k,\vec\varepsilon)
 = P(k,\vec\varepsilon) [1+ (k-k_0)\tilde h(k,\vec\varepsilon)]\,,
 $$
 % ------------- %
where $P(k,\vec\varepsilon)$ is a polynom of order $d$ in the
variable $k$, furthermore, $\lim_{\vec\varepsilon\to 0}
h(k,\vec\varepsilon) = 0$ and $\lim_{\vec\varepsilon\to 0} \tilde
h(k,\vec\varepsilon) = \lim_{\vec\varepsilon\to 0} (k - k_0)^d
h(k,\vec\varepsilon)/P(k,\vec\varepsilon) = 0$. Due to the
fundamental theorem of algebra $P(k,\vec\varepsilon)$ has $d$
zeros, not necessarily different, whose distance from $k_0$
depends continuously on $\vec\varepsilon$. On the other hand, we
have $\forall \delta\,\exists \varepsilon_0' : \forall
\vec\varepsilon \in\mathcal{U}_{\varepsilon_0'}(0), \forall
k\in\mathcal{U}_R(k_0): |\tilde h(k,\vec\varepsilon)|<\delta$ in
view of the above limit relations; choosing then $\delta<1/R$ we
can conclude that zeros of the term $[1+ (k-k_0)\tilde
h(k,\vec\varepsilon)]$ lie outside the ball $\mathcal{U}_R(k_0)$.
 \end{proof}
 % ------------- %

The following lemma slightly generalizes the result to a larger
class of $g(k,\vec \varepsilon)$.
 % ------------- %
 \begin{lemma}\label{rat-lemma2}
Let $(k,\vec \varepsilon)\mapsto F(k,\vec \varepsilon) :
\mathbb{C}\times \mathbb{R}^m \to \mathbb{C}$ be a function
uniformly continuous in $\vec \varepsilon$ for all
$\vec\varepsilon \in \mathcal{U}_{\varepsilon_0}(0)$ and $k \in
\mathcal{U}_R(k_0)$, $\varepsilon_0>0, \,R>0$, and holomorphic in
$k$ in $\mathcal{U}_R(k_0)$ for all $\vec\varepsilon \in
\mathcal{U}_{\varepsilon_0}(0)$. Suppose that $F(k,\vec 0)$ has in
$\mathcal{U}_R(k_0)$ a single zero of multiplicity $d$ at the
point $k_0$; then there exist such $\delta>0$ and $\varepsilon_0'>0$ that
for all $\vec\varepsilon \in \mathcal{U}_{\varepsilon_0'}(0)$ the
sum of the multiplicities of zeros of $F(k,\vec\varepsilon)$ in
$\mathcal{U}_\delta (k_0)$ is equal to $d$.
 \end{lemma}
 % ------------- %
 \begin{proof}
In view of the holomorphy of $F$ and the fact that $F$ has a zero
of order $d$ in $k_0$ one has $F(k,\vec\varepsilon) = (k - k_0)^d
f(k,\vec\varepsilon)$, where $\lim_{\vec\varepsilon \to 0}
f(k,\vec\varepsilon) \not = 0$. Because $f$ is continuous in
$\vec\varepsilon$ we have $f(k,\vec\varepsilon) \not = 0$ for all
$\vec\varepsilon \in\mathcal{U}_{\varepsilon_0'}(0)$,
$k\in\mathcal{U}_R(k_0)$. Hence $f$ does not contribute to zeros
of $F$ in $\mathcal{U}_R(k_0)$ and Lemma~\ref{rat-lemma1} can be
used.
 \end{proof}
 % ------------- %

This conclusion allows us to demonstrate the indicated result. Our
aim is to determine the number of resolvent poles, multiplicity
counting, of the quantum graph with perturbed edge lengths in the
neighbourhood of an original pole of multiplicity $d$. In
particular, we want to find out whether the number of solutions of
the condition (\ref{rat-res}) --- into which we substitute from
(\ref{rat-efu}) --- changes in the neighbourhood of $k_0$. In the
notation of the previous lemma, the function $F$ is given by the
\emph{lhs} of (\ref{rat-res}) and the vector $\vec\varepsilon$
describes the change of the edge lengths.
 % ------------- %
 \begin{theorem}
Let $\Gamma$ be a quantum graph with $N$ finite edges of the
lengths $l_i$, $\:M$ infinite edges, and the coupling described by
the matrix $U = \left(\begin{array}{cc}U_1 &U_2 \\ U_3 & U_4
\end{array}\right)$, where $U_4$ corresponds to the coupling between
the infinite edges. Let $k_0$ satisfy $\mathrm{det}\,
[(1-k_0)U_4 - (1+k_0) I]\not = 0$ and let $k_0$ be a pole of the resolvent $(H -
\lambda \,\mathrm{id})^{-1}$ of a multiplicity $d$. Let
$\Gamma_\varepsilon$ be a geometrically perturbed quantum graph
with the edges of lengths $l_i(1+\varepsilon)$ and the same
coupling as $\Gamma$. Then there exists an $\varepsilon_0 >0$ such that for all
$\vec\varepsilon\in\mathcal{U}_{\varepsilon_0}(0)$ the sum of
multiplicities of the resolvent poles in a sufficiently small
neighbourhood of $k_0$ is $d$.
 \end{theorem}
 % ------------- %
 \begin{proof}
One can rewrite the condition (\ref{rat-con-l2}) for poles of the
resolvent into the form $F(k,\vec\varepsilon) = 0$, where
$\vec\varepsilon$ is the vector of differences of the lengths of
the internal edges. Using the form of the matrices $D_1(k)$ and
$D_2(k)$ and Eq.~(\ref{rat-efu}) one can easily check that if
$\mathrm{det}\,[(1-k_0)U_4 - (1+k_0) I]\not = 0$ then there exists a
neighbourhood $U_R (k_0)$ where $F(k_0,\vec\varepsilon)$ is
holomorphic in $k$ and uniformly continuous in $\vec\varepsilon$,
hence Lemma~\ref{rat-lemma2} can be applied.
 \end{proof}
 % ------------- %

Notice that the condition $\mathrm{det}\,[(1-k_0)U_4
- (1+k_0) I]\not = 0$ is automatically satisfied for $k_0\in
\mathbb{R}^{+}$ because of the inequality $|(k_0+1)/(k_0-1)|>1$ and
the fact that the eigenvalues of $U_4$ do not exceed one in
modulus.

\section*{Acknowledgments}
The research was supported by the Czech Ministry of Education, Youth and Sports within the project LC06002. We thank the referee for suggestions which helped to improve the text.

\end{document}